\documentclass{amsart}
\usepackage{soul,times,amssymb,amsthm,amsmath,amsfonts,bbm,mathrsfs,xcolor,subfigure,graphicx,enumitem,booktabs,euscript}
\allowdisplaybreaks
\usepackage{graphicx}
\usepackage[unicode]{hyperref}
\hypersetup{
    colorlinks,
    linkcolor={blue!80!black},
    citecolor={blue!80!black},
    urlcolor={blue!80!black}
}
\usepackage[normalem]{ulem}
\newcommand\redout{\bgroup\markoverwith{\textcolor{red}{\rule[0.5ex]{2pt}{0.8pt}}}\ULon}
\usepackage[font=footnotesize]{caption}

\newtheorem{theorem}{Theorem}[section]
\newtheorem{lemma}[theorem]{Lemma}

\newtheorem{proposition}[theorem]{Proposition}
\newtheorem{corollary}[theorem]{Corollary}
\theoremstyle{remark}
\newtheorem{remark}{Remark}
\newtheorem{definition}{Definition}[section]
\newtheorem{example}{Example}

\usepackage[top=1.6in,bottom=1.5in,left=1.6in,right=1.6in]{geometry}

\newcommand\nc\newcommand
\nc\bfa{{\boldsymbol a}}\nc\bfA{{\boldsymbol A}}\nc\cA{{\EuScript A}}
\nc\bfb{{\boldsymbol b}}\nc\bfB{{\boldsymbol B}}\nc\cB{{\EuScript B}}
\nc\bfc{{\boldsymbol c}}\nc\bfC{{\boldsymbol C}}\nc\cC{{\mathscr C}}
\nc\bfd{{\boldsymbol d}}\nc\bfD{{\boldsymbol D}}\nc\cD{{\mathscr D}}
\nc\bfe{{\boldsymbol e}}\nc\bfE{{\boldsymbol E}}\nc\cE{{\EuScript E}}
\nc\bff{{\boldsymbol f}}\nc\bfF{{\boldsymbol F}}\nc\cF{{\mathscr F}}
\nc\bfg{{\boldsymbol g}}\nc\bfG{{\boldsymbol G}}\nc\cG{{\EuScript G}}
\nc\bfh{{\boldsymbol h}}\nc\bfH{{\boldsymbol H}}\nc\cH{{\mathcal H}}
\nc\bfi{{\boldsymbol i}}\nc\bfI{{\boldsymbol I}}\nc\cI{{\mathcal I}}
\nc\bfj{{\boldsymbol j}}\nc\bfJ{{\boldsymbol J}}\nc\cJ{{\EuScript J}}
\nc\bfk{{\boldsymbol k}}\nc\bfK{{\boldsymbol K}}\nc\cK{{\EuScript K}}
\nc\bfl{{\boldsymbol l}}\nc\bfL{{\boldsymbol L}}\nc\cL{{\EuScript L}}
\nc\bfm{{\boldsymbol m}}\nc\bfM{{\boldsymbol M}}\nc\cM{{\EuScript M}}
\nc\bfn{{\boldsymbol n}}\nc\bfN{{\boldsymbol N}}\nc\cN{{\EuScript N}}
\nc\bfo{{\boldsymbol o}}\nc\bfO{{\boldsymbol O}}\nc\cO{{\EuScript O}}
\nc\bfp{{\boldsymbol p}}\nc\bfP{{\boldsymbol P}}\nc\cP{{\EuScript P}}
\nc\bfq{{\boldsymbol q}}\nc\bfQ{{\boldsymbol Q}}\nc\cQ{{\EuScript Q}}
\nc\bfr{{\boldsymbol r}}\nc\bfR{{\boldsymbol R}}\nc\cR{{\EuScript R}}
\nc\bfs{{\boldsymbol s}}\nc\bfS{{\boldsymbol S}}\nc\cS{{\EuScript S}}
\nc\bft{{\boldsymbol t}}\nc\bfT{{\boldsymbol T}}\nc\cT{{\EuScript T}}
\nc\bfu{{\boldsymbol u}}\nc\bfU{{\boldsymbol U}}\nc\cU{{\EuScript U}}
\nc\bfv{{\boldsymbol v}}\nc\bfV{{\boldsymbol V}}\nc\cV{{\mathscr V}}
\nc\bfw{{\boldsymbol w}}\nc\bfW{{\boldsymbol W}}\nc\cW{{\mathscr W}}
\nc\bfx{{\boldsymbol x}}\nc\bfX{{\boldsymbol X}}\nc\cX{{\EuScript X}}
\nc\bfy{{\boldsymbol y}}\nc\bfY{{\boldsymbol Y}}\nc\cY{{\mathscr Y}}
\nc\bfz{{\boldsymbol z}}\nc\bfZ{{\boldsymbol Z}}\nc\cZ{{\EuScript Z}}
\nc{\1}{{\mathbbm{1}}}
\nc\reals{{\mathbb R}}
\nc{\ff}{{\mathbb F}}
\nc{\PP}{{\mathbb P}}
\newcommand{\cl}{\mathrm{cl}}

\nc{\remove}[1]{}

\DeclareSymbolFont{bbold}{U}{bbold}{m}{n}
\DeclareSymbolFontAlphabet{\mathbbold}{bbold}

\usepackage{xargs}
\usepackage[colorinlistoftodos,prependcaption]{todonotes}

\usepackage{todonotes}

\newcommandx{\rednote}[2][1=]{\todo[linecolor=red,backgroundcolor=red!25,bordercolor=red,#1]{#2}}
\newcommandx{\bluenote}[2][1=]{\todo[linecolor=blue,backgroundcolor=blue!25,bordercolor=blue,#1]{#2}}
\newcommandx{\yellownote}[2][1=]{\todo[linecolor=yellow,backgroundcolor=yellow!25,bordercolor=yellow,#1]{#2}}
\newcommandx{\greennote}[2][1=]{\todo[inline,linecolor=olive,backgroundcolor=green!25,bordercolor=olive,#1]{#2}}

\newcommand{\N}{{\mathbb N}}

\newcommand{\R}{{\mathbb R}}
\newcommand{\Z}{{\mathbb Z}}

\newcommand{\ccap}{\mathsf{cap}}
\newcommand{\mathset}[1]{\left\{#1\right\}}
\newcommand{\limsupup}[1]{\limsup_{#1\rightarrow\infty}}

\title{Storage codes and recoverable systems on lines and grids }

\author[]{Alexander Barg}\thanks{Alexander Barg is with ISR/Dept. of ECE, University of Maryland, College Park, MD, USA. 
Email: abarg@umd.edu. His research was supported by NSF grant CCF2104489 and NSF-BSF grant CCF2110113.}
\author[]{Ohad Elishco}\thanks{Ohad Elishco is with Ben-Gurion University of the Negev, Israel, Email: ohadeli@bgu.ac.il. 
His research was supported by NSF-BSF grant CCF2020762.}
\author[]{Ryan Gabrys}\thanks{Ryan Gabrys is with the University of California at San Diego, CA, USA. Email: ryan.gabrys@gmail.com.}
\author[]{Geyang Wang}\thanks{Geyang Wang is with ISR/Dept. of ECE, University of Maryland, College Park, MD, USA. 
Email: wanggy@umd.edu. His research was supported by NSF-BSF grant CCF2110113.}
\author[]{Eitan Yaakobi}\thanks{Eitan Yaakobi is with the Technion - Israel Institute of Technology, Haifa, Israel. Email: yaakobi@cs.technion.ac.il.
}

\begin{document}
\begin{abstract}
A storage code is an assignment of symbols to the vertices of a connected graph
$G(V,E)$ with the property that the value of each vertex is a function of the
values of its neighbors, or more generally, of a certain neighborhood of the vertex in $G$.
In this work we introduce a new construction method of storage codes, enabling one to construct new codes from known ones via an interleaving procedure driven by resolvable designs. We also study storage codes on $\Z$ and $\Z^2$ (lines and grids), finding
closed-form expressions for the capacity of several one and two-dimensional systems depending on their recovery set, using connections between storage codes, graphs, anticodes, and difference-avoiding sets.
\end{abstract}
\maketitle

\section{Introduction}
The concept of locality in coding theory has been the subject of extensive research during the last 
decade. In general terms, we are interested in encoding information in a way that permits recovery of 
a single bit of the data by accessing a small number of the available bits (symbols) in the encoding.
Local recovery has been studied in the scenarios of coding for error correction (channel coding) 
\cite{gopalan2011locality} and source coding \cite{MazumdarWornell2015}. The concept of locality was further extended to account for graph constraints in the system, whereby the symbols of the encoding are placed on the vertices of a graph.

The main problem concerning storage 
codes is establishing the largest size of $\cC$ for a given class of graphs. It was soon realized 
that this problem can be equivalently phrased as the smallest rate of symmetric index codes \cite{BarYossef2011}, or the largest success probability in (one variant of) guessing games on the
graph $G$ \cite{Riis2007}; see \cite{AK2015,BargZemor2021} for a more detailed discussion. Some high-rate storage codes were recently constructed in 
\cite{barg2,BargZemor2021,GolovnevHaviv2022}.

The concept of storage codes was recently extended to infinite graphs \cite{ElishcoBarg2020} under the
name of {\em recoverable systems}.
For a finite set of integers $R$, we write $n+R:=\mathset{n+r ~:~ r\in R}$. For a sequence $x\in \cQ^{V}$ and for a set of integers $R$, denote by $x_R$ the restriction of $x$ to the positions in $R$. The authors of \cite{ElishcoBarg2020} considered codes $X$ formed of bi-infinite sequences $x\in \cQ^\Z$ with the property that for any $i\in \Z$ the value $x_i$ is found as $f(x_{i+R}),$ 
$R=\mathset{j:0<|j|<l}$ $(l\ge 1)$
and $f:\Sigma^{2l}\to \cQ$ is a deterministic function, independent of $i$. They additionally assumed that the system $X$ is shift invariant, i.e., if $x\in X$ then also a left shift $Tx\in X$ ($T$ acts by shifting all the symbols in $x$ one place to the left). This assumption enabled them to rely on methods from constrained systems to estimate the growth rate of the set of allowable sequences, or the capacity of recoverable systems.

The object of this paper is storage capacity of finite graphs, constructions of storage codes, as well 
as capacity and constructions of recoverable systems. For finite graphs we present a construction
of codes based on interleaving known codes controlled by resolvable designs to obtain new codes from known ones. If the seed codes are optimal, then so are the interleaved ones. For graphs with 
transitive automorphisms, we phrase the capacity problem in terms of the code-anticode bound \cite{Delsarte1973,AAK2001}. 

For the infinite case, using finite subgraphs of $\Z$ and $\Z^2$, we obtain capacity values for certain recovery regions $R$. For instance, for $\Z_2,$ we find storage capacity for balls in the $l_1$ and $l_\infty$ metrics as well as certain cross-shaped regions. These results do not involve the shift invariance assumption.

Some of the results of this paper were presented at the 2022 IEEE International Symposium on Information Theory and published as an extended abstract \cite{BEGY2022}. Here we add new results, notably Section~\ref{sec:interleaving}, and also provide complete or corrected proofs of the results announced in  \cite{BEGY2022}.

\section{Preliminaries}\label{sec:p}

\subsection{Storage codes for finite graphs}\label{sec:finite}
\begin{definition} {\cite{Mazumdar2015,Shanmugam2014}} Let $x\in\cQ^n$ be a word over a finite alphabet $\cQ$ and let $G=(V,E)$ be a finite
graph with $|V|=\{1,\dots,n\}$ and a fixed ordering of the vertices. We say that $x$ is assigned to $G$ if there is a bijection $\{1,\dots,n\}\to V$ which places entries of $x$ on the vertices.  
A {\em storage code} $\cC$ on $G$ is a collection of assignments of words such that
for each $v\in V$ and every $x\in\cC$, the value $x_v$ is a function of $\{x_u, u\in \cN(v)\}$, where $\cN(v):=\{u:(v,u)\in E(G)\}$ is the vertex neighborhood of $v$ in $G$.
\end{definition}

We briefly mention the known results for storage codes on finite graphs as defined in the opening paragraph of the paper. Let $G$ be a graph with $n$ vertices and let $\cR_q(\cC):=\frac 1n\log_q{|\cC|}$ be the rate of a $q$-ary code $\cC$ on $G$. In this case, the recovery set of a vertex $v$ is $\cN(v),$ and it depends on $v$. 
We denote the largest attainable rate of a storage code on $G$ by $\ccap(G):= \sup_q \cR_q(\cC)$. 

There are several ways of constructing storage codes with large rate. The most well-known one is the 
edge-to-vertex construction: given a ($d$-regular) graph, place a $q$-ary symbol on every edge and assign each vertex a $d$-vector of symbols written on the edges incident to it (again we assume an ordering of the edges). The size of the code is $(q^d)^{n/2},$ resulting in the rate value $1/2$ irrespective of the value of $q$. Here is an example with $q=d=2$:

\begin{center}\begin{tikzpicture}[rotate=-18,scale=1.1]
\node (v1) at ( 0:1) [draw,shape=circle] {$X_1$};
\node (v2) at ( 72:1) [draw,shape=circle] {$X_2$};
\node (v3) at (2*72:1) [draw,shape=circle] {$X_3$};
\node (v4) at (3*72:1) [draw,shape=circle] {$X_4$};
\node (v5) at (4*72:1) [draw,shape=circle] {$X_5$};
\draw (v1) -- (v2) -- (v3)-- (v4)-- (v5)--(v1);
\node at (0.8,0.55) {$a$};
\node at (0.8,-0.55) {$b$};
\node at (-0.3,-1) {$c$};
\node at (-1,0.0) {$d$};
\node at (-0.3,0.95) {$e$};
\node at (0.1,1.6) {$\begin{matrix} \textcolor{red}{e}\\[-.05in]\textcolor{blue}{a} \end{matrix}$};
\node at (1.6,0) {$\begin{matrix} \textcolor{red}{a}\\\textcolor{blue}{b} \end{matrix}$};
\node at (0.5,-1.7) {$\begin{matrix} \textcolor{red}{b}\\[-.05in]\textcolor{blue}{c} \end{matrix}$};
\node at (-1.4,-0.8) {$\begin{matrix} \textcolor{red}{c}\\\textcolor{blue}{d} \end{matrix}$};
\node at (-0.95,1.3) {$\begin{matrix} \textcolor{red}{d}\\[-.05in]\textcolor{blue}{e} \end{matrix}$};
\end{tikzpicture}
\end{center}

\vspace*{-.1in} This method extends to the case when every vertex is incident to the same number of cliques. For instance, if this number is one, then the graph can be partitioned into $k$-cliques. To construct a code,
we put a single parity symbol on every clique and distribute the symbols of the parity to the vertices that form it, resulting in rate $\cR=(k-1)/k$. A further extension, known as {\em clique covering} \cite{BarYossef2011}, states that 
\begin{equation}\label{eq:clique}
\ccap(G) \geq 1 - \alpha(G)/n,  
\end{equation}
where $\alpha(G)$ is the smallest size of a clique covering in $G$.
Among other general constructions, we mention the {\em matching construction}, which yields
$$\ccap(G)\ge M(G)/n,$$
where $M(G)$ is the size of the largest matching in $G$. 

Turning to upper bounds, we note a result in \cite{BarYossef2011}, Theorem 3 (also \cite{Mazumdar2015}, Lemma 9), which states that
\begin{equation}\label{eq:IB}   
\ccap(G)\le 1-\gamma(G)/n,
\end{equation}
where $\gamma(G)$ is the independence number of $G$, i.e., the size of the largest independent set of vertices (in other words, the size of the smallest vertex cover of $G$). 

Given a graph $G$, we may sometimes want to look at neighbors at distance more than one from the failed vertex to recover its value, such as the set $R$ discussed in Def.~\ref{def:main}. Effectively, this changes the 
connectivity of the graph, so while we keep the same set of vertices $V$, the edges are now drawn according
to where the failed vertex collects the data for its recovery. Assuming that the recovery region can be defined consistently for all the vertices, we will use the notation $G_R$ in our discussion of storage 
codes for this case. This agreement will be used in particular when we address the two-dimensional grid $\Z\times\Z$ below.

Since $R$ does not have to be symmetric, the graph $G_R$ generally is directed (as is often the case in index coding \cite{BarYossef2011,AK2018}). In this case, bound \eqref{eq:IB} affords a generalization, which we proceed to describe. 
For a subset of vertices $U\subseteq V$ in a graph $G=(V,E)$, denote by $G(U)$ the induced subgraph. A directed graph with no directed cycles is called a {\em directed acyclic graph} (DAG). It is known that a graph is a DAG if and only if it can be topologically ordered \cite[p.~258]{Knuth1973}, i.e., there is a numbering of the vertices such that the tail of every arc is smaller than the head. A set of vertices $S\subseteq V$ in a graph $G$ is called a \emph{DAG set} if the
subgraph induced by $S$ is a DAG. 
With this preparation, the following {\em maximum acyclic induced subgraph (MAIS) bound} is true; see~\cite{BarYossef2011}, Theorem 3, or \cite{AK2018}, Sec. 5.1.
\begin{theorem}\label{th:DAG_set} 
Let $G=G(V,E)$ be a graph and let $\delta(G)$ be the size of the largest DAG set in it. Then,
   \begin{equation}\label{eq:MAIS}
   \ccap(G) \leq 1 - \delta(G)/n.
   \end{equation}
\end{theorem}
\begin{proof} 
  Let $S \subset V$ be a DAG set. We argue that the values stored on $S$ are uniquely determined by the values stored on $G \setminus S$. 
  First, consider the topological ordering of $S$, all the outgoing edges of the last ordered $s \in S$ end in $G\setminus S$.
  Therefore, $s$ can be recovered by $G \setminus S$.
  Then we can recover the second last vertex of $S$ by $s$ and  $G \setminus S$. 
  Proceeding like this, we are able to recover all values stored on $S$. 
  Taking $S$ to be the largest DAG set in $G$, we have $\ccap_q(G) \leq 1 - \delta(G)/n$ for all $q \ge 2$. 
\end{proof}

The bounds~\eqref{eq:IB} and~\eqref{eq:MAIS} are tight and can be achieved by clique partition codes.
A \emph{clique partition} of a graph $G$ is a partition of the graph into subsets such that the induced graph by every subset is a clique. 

Other upper bounds on the rate of storage codes are found in \cite{MazMcgVor2019} for the symmetric case and in \cite{AK2018} for the general case. In this paper we will need a {\em linear programming bound} for the capacity of storage codes proved in \cite{MazMcgVor2019} (its statement is somewhat technical and is given in Appendix \ref{sec:lp-bound}).  We note that the authors of \cite{MazMcgVor2019} also found some families of codes that achieve it.

\subsection{Capacity for finite graphs and the code-anticode bound}\label{sec:code-anticode}

Given a finite graph $G(V,E)$, we define the graphical distance $\rho(u,v)$ as the length of the shortest path in $G$
between $u$ and $v$. A code in $G$ is a subset $\cC\subset V$, and it is said to have minimum distance $d$
if $\rho(u,v)\ge d$ for all pairs of distinct vertices $u,v\in \cC$, or, in other words, if for any two distinct 
code vertices $u,v\in \cC$ the balls $\cB_r(u)$ and $\cB_r(v)$ of radius $r=\lfloor (d-1)/2\rfloor$ are disjoint.
A code $\cC\subset V$ is called an $r$-covering code if $\bigcup_{v\in \cC} \cB_r(v)=V.$ We denote by 
$C(G,r)$ the smallest size of an $r$-covering code in the graph $G.$ 

Suppose that we are given a finite set $V, |V|=n$ and a subset $R(v)\subset V\backslash\{v\}$ that serves the recovery region for the vertex $v$. In this section the recovery regions will be given by balls of a
given radius $r$ in the metric $\rho$, the same for every vertex $v\in V.$ We construct a graph $G_r=G(V,E)$ by connecting each vertex $v$ with all the vertices $u$ such that $1\le \rho(v,u)\le r$. 
For instance, in Sec.~\ref{sec:ell-metrics} below, $V$ will be an $n\times n$ region of $\Z^2$ and $R(v)=\cB_r(v)\backslash\{v\}$, where $\cB_r(\cdot)$ is a ball of radius $r$ in some metric on $\Z^2$ (we focus on the $l_1$ and $l_\infty$ distances).

Denote by $A(G;r+1)$ the size of the largest code in $G$ with minimum distance at least $r+1$. Any such a code is an independent set in $G_r,$ and thus by \eqref{eq:IB}
\begin{equation}\label{eq:min_dist_bound}
\ccap(G_r) \leq 1 - \frac 1{n}\gamma(G_r) = 1-\frac 1{n} A(G;r+1).
\end{equation}
This gives an interpretation of the bound \eqref{eq:IB} in coding-theoretic terms.
Furthermore, every ball of radius $\lfloor r/2\rfloor$ in $G$ is a clique in the graph $G_r$, and thus, any $\lfloor r/2\rfloor$-covering code in $G$ can form a clique covering of the graph $G_r$. As above, let $\alpha(G_r)$ be the number of cliques in the smallest covering. 
From \eqref{eq:clique} we obtain that
 \begin{equation*}\label{eq:spb}
\ccap(G_r) \geq 1- \frac1{n}\alpha(G_r) \geq 1- \frac1{n}C\left(G; \left \lfloor \frac{r}{2} \right \rfloor \right). 
\end{equation*}

Assume the number of vertices in the ball of radius $\lfloor(r-1)/2\rfloor$ in $G$ does not depend on the center, and denote this number by $B_G(r)$. There are general conditions for this to hold, for instance if the graph $G$ is arc-transitive. Then the {\em sphere packing bound} implies that $A(G;r+1)\leq \frac{n}{B_G(\lfloor(r-1)/2\rfloor)}$. If $r$ is even and there exists a perfect $r/2$-error-correcting code then $A(G;r+1) = C(G;r/2) = \frac{n}{B_G(\lfloor(r-1)/2\rfloor)}$ and so 
\begin{align*}
\ccap(G_r) = 1 -\frac1{n} A(G;r+1) = 1 - \frac1{B_G(\lfloor(r-1)/2\rfloor)}. 
\end{align*}

These relations can be used to derive capacity bounds for graphs. Below we use an extension of the sphere-packing bound, 
known as the {\em code-anticode bound}, to derive exact values of capacity for some recoverable systems in $\Z^2$.
 A subset of 
vertices in $G$ with diameter $D$ is called an {\em anticode}. For instance, a ball of radius $\tau$ is an anticode with 
diameter $2\tau$. It is known that in many cases the largest size of an anticode of even diameter $D$ is achieved by a ball of radius $D/2$. For odd values of $D$ the largest anticode is usually constructed by taking a union of two  balls of radius $(D-1)/2$ whose centers are adjacent in $G.$  

Delsarte \cite[Thm.3.9]{Delsarte1973} proved that if $G$ contains a code $\cC$ with 
minimum distance $r+1$ and an anticode $\cD$ of diameter $r$, {\em and $G$ is distance-regular}, then 
    \begin{equation}\label{eq:c-a}
       |\cC||\cD|\le n.
    \end{equation} 
Taking $\cD$ a ball of radius $\lfloor(r-1)/2\rfloor$ recovers the sphere-packing bound. The condition of distance regularity was relaxed in \cite{AAK2001,E11}. In particular, \cite[Thm.~$1'$]{AAK2001}
implies that \eqref{eq:c-a} holds true as long as $G$ admits a transitive automorphism group. A code that satisfies the code-anticode bound with equality is called {\em diameter perfect}. The existence of diameter perfect codes is in general a difficult question; see e.g., \cite{shi2022family,etzion2022perfect} for recent references.
In our examples, diameter perfect codes will exist, and they will also 
generate a tiling of the graph with congruent copies of the corresponding
anticode. 

We continue with the following general claim.
\begin{theorem}\label{cor:covering} Suppose that $G$ has a transitive automorphism group and contains a diameter perfect code.
Then,
    \begin{equation}\label{eq:lb}
    \ccap(G_r) = 1 - \frac{1}{D_G(r)},
    \end{equation}
 where $D_G(r)$ is the largest possible size of an anticode in $G$ of diameter $r$.   
\end{theorem}
\begin{proof}
An anticode $\cD$ of diameter $r$ forms a clique in the graph $G_r$ because the recovery region of every vertex in $\cD$ includes all the other vertices in $\cD.$ Thus, by assumption, there is a clique covering of $G_r$, and if the anticodes are of the largest possible size, this clique covering is a smallest one. Then from \eqref{eq:clique} we conclude that $\ccap(G_r) \ge 1 - \frac{1}{D_G(r)}$. 

Further, from \eqref{eq:c-a} we have $A(G;r+1)\le n/D_G(r)$, and since $G$ contains a diameter perfect code, this is 
in fact an equality.  Then \eqref{eq:min_dist_bound} implies that $\ccap(G_r) \le 1 - \frac{1}{D_G(r)}$, completing the proof.
\end{proof}

As an example, consider a {\em discrete torus}, i.e., a graph $\cT_n$ on the vertex set
$V=[n]\times [n]$ with an edge between $(i_1,j_1)$ and $(i_2,j_2)$ whenever $(i_1-i_2,j_1-j_2)$ equals one of $(1,0),(0,1),(-1,0),$ $(0,-1)$ modulo $n$. Consider storage codes on $\cT_n$ based on the recovery sets formed by the entire rows and columns (circles) on the torus. In other words, we take the recovery set in the form $R:=(0,\ast)\cup(\ast,0)\backslash(0,0)$ where the unspecified coordinates are allowed to vary over the entire set $\Z_n.$ The graph $(\cT_n)_R$ representing the system is obtained from
$\cT_n$ by connecting all pairs of vertices whose coordinates are identical
in either the first or the second position. 

\begin{theorem}
For $n \geq 3$ the storage capacity of the discrete torus $G:=(\cT_n)_R$ is
\begin{align}
    \ccap(G) = 1-\frac{1}{n}.\label{eq:torus-bound}
\end{align}
\end{theorem}
\begin{proof}
 Notice first that for any two vertices $(i_1, j_1)$, $(i_2, j_2)$, there is an edge between $(i_1,j_1)$ and $(i_2, j_2)$ if and only if $d_H\left( (i_1, j_1), (i_2, j_2) \right) = 1$, where $d_H$ denotes the Hamming distance. Placing a parity constraint on every row of $\cT_n$ 
yields a code of rate $1-1/n,$ proving a lower bound in \eqref{eq:torus-bound}.
In order to show that it is also an upper bound, first observe that the automorphism group of $G$ acts transitively on it. On account of Theorem~\ref{cor:covering}, to complete the proof it suffices to show that $D_{G}(1) = n.$ In words, we want to show that the largest anticode of diameter $1$ in the graph $G$ under the Hamming metric is of size $n$, which is immediate.
\end{proof}

\begin{remark} Using the results of Theorem~\ref{thm:1-infty}, we can also find capacity of the torus with recovery
region defined by the $l_1$ or $l_\infty$ distance.
\end{remark}

\section{Interleaving Structures}\label{sec:interleaving}

In this section, we introduce a way to construct new storage codes for finite graphs by interleaving existing ones. 
Let $\cC$ be a storage code on a finite graph $G$ over the alphabet $Q$. Let $x^1,\dots,x^{ks}$ be arbitrary $ks$ codewords from $\cC$, where $k,s \in \N$. 
Roughly speaking, the interleaving operation maps these $ks$ codewords to $(Q^k)^{ns}$, and the set of all interleaved codewords forms a storage code on a new graph, defined below as a part of the code construction. Denote by $\bar{\cC}$ the interleaved code and $\bar{G}$ the corresponding graph.
We will define the interleaving operation such that $\cR_{q^k}(\bar\cC) = \cR_q(\cC)$, thus $\ccap_{q^k} (\bar{G})\ge \ccap_q(G)$ holds.
Moreover, if $\cR_q(\cC)$ meets the MAIS bounds~\eqref{eq:MAIS} or the linear programming bound of Mazumadar et al.~\cite{MazMcgVor2019}, then we have $\cR_{q^k}(\bar\cC) = \ccap(\bar{G})$, enabling one to construct optimal storage codes for a wide range of parameters.

\subsection{Interleaving construction}
\begin{definition} \label{def: orthogonal partition}
Let $\cM$ be a set of $k \times s$ matrices ($s \ge k$) over $\N$. We say that $\cM$ 
forms a family of orthogonal partitions of the set of integers $\{1,\dots,ks\}$ if
    \begin{enumerate}
        \item Every $A \in \cM$ contains every element of $\{1,\dots,ks\}$ exactly once, thus, the columns of $A$ form a partition of $\{1,\dots,ks\}$.
        \item For every pair of distinct matrices $A, B \in \cM$, every column of $A$ has common entries with $k$ columns of $B$, and every such pair of intersecting columns has exactly one common element.
    \end{enumerate}
\end{definition}
We call $(k, s)$ the shape and $|\cM|$ the size of the family of orthogonal partitions,  respectively.

\begin{example}\label{eg:Kirkman}
Let $k=3,s=5.$ The following set of matrices forms a family of orthogonal partitions of the
set $\{1,2,\dots,15\}$:
 \begin{align*}
       &   \left[\!\begin{array}{*5{c@{\hspace*{0.05in}}}}
        1 & 4 & 5 & 6 & 7 \\
        2 & 10 & 8 & 9 & 11 \\
        3 & 14 & 13 & 15 & 12 \\
       \end{array}\right],
       \left[\!\begin{array}{*5{c@{\hspace*{0.05in}}}}
        1 & 2 & 3 & 4 & 6 \\
        8 & 5 & 13 & 11 & 10 \\
        9 & 7 & 14 & 15 & 12 \\
       \end{array}\right],
       \left[\!\begin{array}{*5{c@{\hspace*{0.05in}}}}
        1 & 2 & 3 & 4 & 7 \\
        10 & 13 & 5 & 8 & 9 \\
        11 & 15 & 6 & 12 & 14 \\
       \end{array}\right],
                   \left[\!\begin{array}{*5{c@{\hspace*{0.05in}}}}
        1 & 2 & 3 & 6 & 7 \\
        4 & 12 & 9 & 11 & 8 \\
        5 & 14 & 10 & 13 & 15 \\
       \end{array}\right], 
       \\
      &  \left[\!\begin{array}{*5{c@{\hspace*{0.05in}}}}
        1 & 2 & 3 & 4 & 5 \\
        6 & 8 & 12 & 9 & 11 \\
        7 & 10 & 15 & 13 & 14 \\
       \end{array}\right],
       \left[\!\begin{array}{*5{c@{\hspace*{0.05in}}}}
        1 & 2 & 3 & 5 & 6 \\
        12 & 9 & 4 & 10 & 8 \\
        13 & 11 & 7 & 15 & 14 \\
       \end{array}\right],
       \left[\!\begin{array}{*5{c@{\hspace*{0.05in}}}}
        1 & 2 & 3 & 5 & 7 \\
        14 & 4 & 8 & 9 & 10 \\
        15 & 6 & 11 & 12 & 13 \\
       \end{array}\right].
  \end{align*}
  For instance, the second column of the first matrix interests with columns 3,4, and 5 of the second matrix, etc.
\end{example}

Families of orthogonal partitions can be constructed relying on resolvable designs~\cite{CollbornHand2007}, which we define next. Let $\cB$ be a collection of $k$-subsets
of a finite set $V$. We call the elements of $V$ and $\cB$ points and blocks, respectively.
The pair $(V,\cB)$ is called an $r$-$(v,k,\lambda)$ block design if every $r$-subset of $V$ is contained in exactly $\lambda$ blocks. A set of blocks is called a parallel class if they form a partition of $V$. Finally, a block design is called resolvable if its set of blocks $\cB$ can be partitioned into parallel classes.

\begin{proposition}\label{prop:cons_ort_par}
A $2$-$(v,k,1)$ resolvable design defines a family of orthogonal partitions of shape $(k,s)$ and size $(v-1)/(k-1)$, where $s = v/k$.
\end{proposition}
\begin{proof}
Given a $2$-$(v,k,1)$ resolvable design, each parallel class yields a matrix in the family of orthogonal partitions, whose columns are the blocks in the class\footnote{We only require each column to contain all the elements from the block and do not impose any ordering.}.
Therefore, every matrix contains all the points exactly once, and two columns from different matrices are either disjoint or intersect on one point. Indeed, if two columns intersect on two points, this would imply that a $2$-subset is contained in two blocks, which is a contradiction. 
\end{proof}

By Proposition~\ref{prop:cons_ort_par}, the existence of $2$-$(v,k,1)$ resolvable design implies the existence of orthogonal partition families.
The necessary conditions for the existence of $2$-$(v,k,1)$ resolvable designs are $k \vert v$ and $(k-1) \vert (v-1)$. 
If $v$ and $k$ are both powers of the same prime, then the necessary conditions are also sufficient ~\cite[Theorem 7.10]{CollbornHand2007}.
In particular,  for every $v \equiv 3 \mod 6$ there exists a family of $2$-$(v,3,1)$ resolvable designs called Kirkman triple systems (Example~\ref{eg:Kirkman} is obtained from a Kirkman triple system with $v=15$). Therefore, resolvable designs give us a rich collection of orthogonal families\footnote{Every $2$-$(v,k,1)$ resolvable design yields $(k!)^{\frac{v(v-1)}{k(k-1)}} (s!)^{\frac{v-1}{k-1}}$ orthogonal families. To see this, note that there are $v(v-1)/k(k-1)$ blocks and $(v-1)/(k-1)$ parallel classes, and we can rearrange elements in each column and permute the columns of each matrix.}. 

Now we will present a way of obtaining new storage codes from known codes. Given a code $\cC$ on a graph $G=(V,E)$ and a family $\cM$ of orthogonal partitions of shape $(k,s)$, we will first construct a new graph $\bar{G} = (\bar{V},\bar{E}).$
Fix a coloring of $G$ in $c$ colors such that adjacent vertices are assigned
different colors. 
Pick $c$ matrices from  $\cM$ and label them with colors $M_1, \dots, M_c$.
Let $\bar V = \{1,\dots,n\} \times \{1,\dots,s\} $, and $((t,\mu),(t',\mu')) \in \bar{E}$ if and only if $(t,t') \in E$, and column $\mu$ of $M_{c(t)}$ intersects column $\mu'$ of $M_{c(t')}$, where $c(\cdot)$ denotes the color of the vertex. 
Note that each vertex $t \in V$ corresponds to an independent set $\bar{t} = \{(t,1), \dots, (t,s)\} \in \bar{V}$, and there are edges between vertices in $\bar{t}$ and $\bar{t}'$ only if $(t, t') \in E$. This construction works for
both undirected and directed graphs $G$.

Next, given a storage code $\cC$ on a finite graph $G = (V,E)$ over the alphabet $\cQ$, $|\cQ| = q$, and $V = \{1,\dots,n\}$, we define an interleaving procedure that produces a storage code $\bar\cC$ on $\bar G$ over $(\cQ^k)^{ns}.$

\begin{enumerate}
  \item Choose codewords $x^1, x^2,\dots,x^{ks}$ from $\cC$ (not necessarily distinct).
For each $x^\lambda \in \cC$, $\lambda = 1,\dots, ks$, denote by $x^\lambda_v$ the symbol stored on the vertex $v \in V$.
  \item For each vertex $t \in V$, denote by $r = c(t)$ its color, and form a $k \times s$ matrix such that entry $(i,j)$ is $x^{M_r(i,j)}_t$. In other words, we arrange $x^1_t, \dots, x^{ks}_t$ in a $k \times s$ matrix $X_t$ according to $M_r \in \cM$, namely, the $i,j$-th entry is $x^{M_r(i,j)}_t$.
  \item A codeword $\bar{x} \in (\cQ^k)^{ns}$ is obtained by concatenating the columns of the matrices $X_1, X_2, \dots, X_n$ in some fixed order.
  Specifically, for $\mu\in\{1,\dots,s\}$ and $t\in\{1,\dots,n\}$ we assign column $\mu$ of $X_t$ to the vertex $(t,\mu)$.
 \item The collection of codewords $\bar{x}$ constructed in the previous steps forms
 the code $\bar{\cC}$.
\end{enumerate}

\begin{proposition}\label{prop:interleaved}
  The interleaved code $\bar{\cC}$ is a storage code on $\bar{G}$.
\end{proposition}

\begin{proof}
  Suppose that node $(t,\mu)$ is erased, and denote by $(x_t^{\lambda_1},\dots, x_t^{\lambda_k})^T$ the column stored on it, where $x^{\lambda_1},\dots ,x^{\lambda_k} \in C$ are the codewords that were interleaved to define the codeword $\bar x$.
  Note that $x_t^{\lambda_j}$ is a function of $\{x_{t'}^{\lambda_j}, t' \in \cN(t)\}$, and $x_{t'}^{\lambda_j}$ is stored on exactly one of the nodes $\bar{t'}$, and by definition of orthogonal partitions, $x_{t'}^{\lambda_j}, x_{t'}^{\lambda_{j'}}$ are stored on different vertices if $j \ne j'$.
  In a nutshell, $(t,\mu)$ can be recovered from the values $\bigcup_{j=1}^k \{x_{t'}^{\lambda_j}, t' \in \cN(t)\}$, which are stored on exactly $k |\cN(t)|$ vertices of $\bar{G}$.
\end{proof}

\begin{remark}
  It is straightforward to check that the described method enables us to interleave codewords from different storage codes on $G$ by the same method. 
Namely, let $\cC_1, \dots, \cC_{ks}$ be $ks$ storage codes on $G$, and 
choose codewords $x^i \in \cC_i$ for each $i=1,2,\dots,ks$. They can be used in the above procedure to define a new code. At the same time, if our goal is to construct
large-size codes, then the construction should rely on codes of the largest 
known size, for instance, copies of the same code.
\end{remark}

\begin{example}[Example of interleaving]
  Let $G$ be a triangle and $\cC$ be a linear parity-check storage code over $\mathbb{F}_3$.
  \begin{enumerate}
    \item Color $G$ by $3$ colors and let 
    \[
      \cM = \left \{ 
        M_1 = \begin{bmatrix}
        1 & 3 & 5 \\
        2 & 4 & 6
      \end{bmatrix}, 
      M_2 = \begin{bmatrix}
        1 & 2 & 6 \\
        5 & 3 & 4 
      \end{bmatrix},
      M_3 = \begin{bmatrix}
        1 & 5 & 6 \\
        4 & 2 & 3
      \end{bmatrix}
      \right \}.
      \]

      \item Pick $6$ codewords from $\cC$, denote them by $x^1, \dots, x^6$, and arrange them into matrices $X_1, X_2, X_3$ relying on the orthogonal partitions given by $M_1, M_2, M_3$.
      We obtain
      \begin{align*}
        X_1 = \begin{bmatrix}
          x_1^1 & x_1^3 & x^5_1 \\
          x_1^2 & x_1^4 & x_1^6 
        \end{bmatrix}, 
        X_2 = \begin{bmatrix}
          x_2^1 & x_2^2 & x_2^6 \\
          x_2^5 & x_2^3 & x_2^4
        \end{bmatrix}
        X_3 = \begin{bmatrix}
          x^1_3 & x_3^5 & x_3^6 \\
          x_3^4 & x_3^2 & x_3^3 
        \end{bmatrix},
      \end{align*}
      and 
      $$
      \bar{x} = \begin{pmatrix}
        x_1^1 & x_1^3 & x^5_1 & x_2^1 & x_2^2 & x_2^6 & x^1_3 & x_3^5 & x_3^6 \\
        x_1^2 & x_1^4 & x_1^6 & x_2^5 & x_2^3 & x_2^4 & x_3^4 & x_3^2 & x_3^3 
      \end{pmatrix}.
      $$
  \end{enumerate}
  For instance, choosing $\{x^1 = 111, x^2 = 222, x^3 = 000, x^4 = 120, x^5 = 012, x^6 = 102\}\subset \cC$, we obtain
  \[
  \begin{split}
    \bar{x} & =
    \begin{pmatrix}
    1 & 0 & 0 & 1 & 2 & 0 & 1 & 2 & 2 \\
    2 & 1 & 1 & 1 & 2 & 2 & 0 & 0 & 0
  \end{pmatrix}.
  \end{split}  
  \]
  The corresponding graph $\bar{G}$ defined on $[3] \times [3]$ is shown in the figure.
  \begin{figure}[h]
    \includegraphics[width=\textwidth]{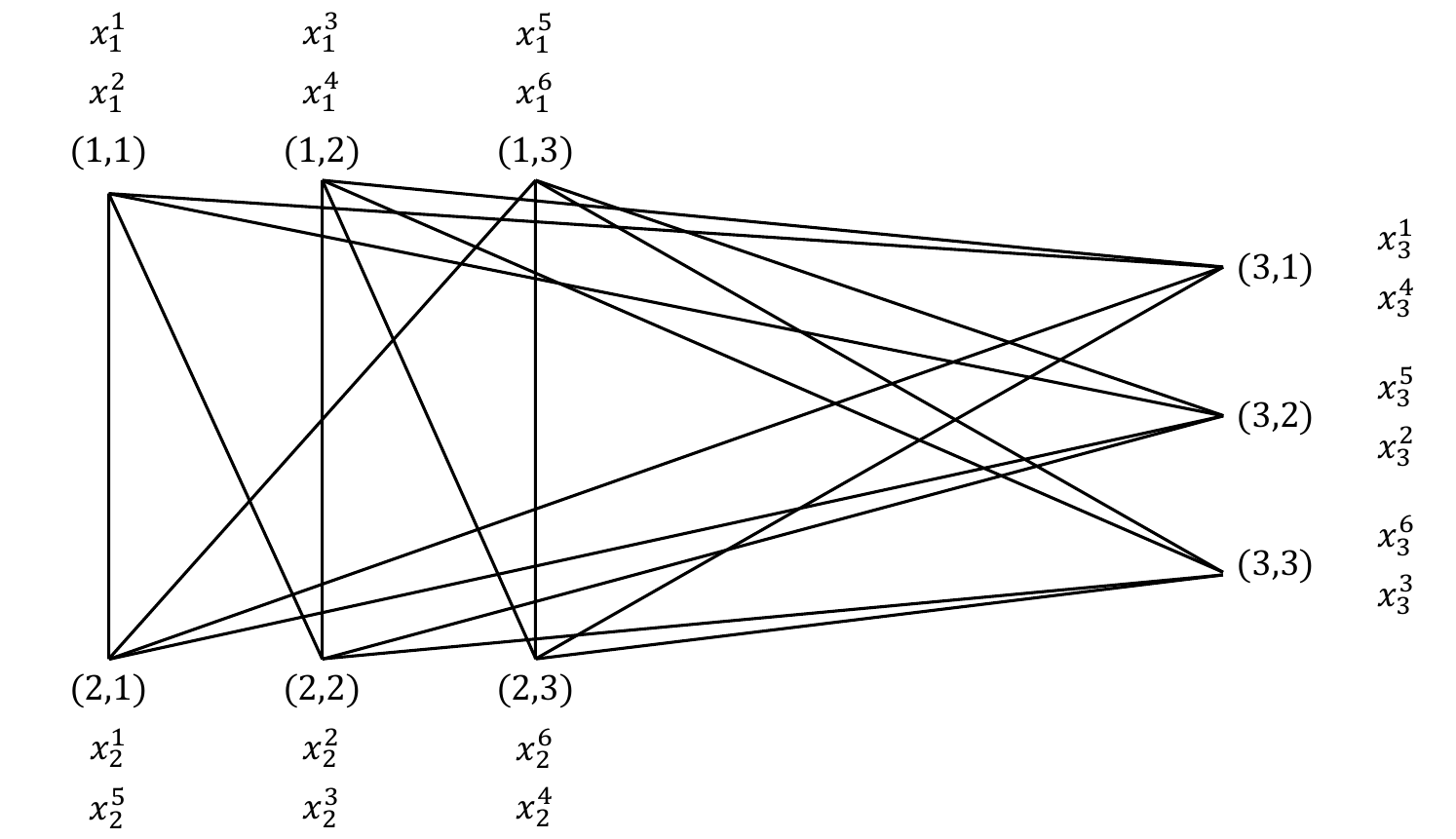}
    \caption{The figure shows the graph $\bar G$ in the example. We have $n=3,k=2,s=3$,
    and $c=3$. The graph has 9 vertices 
    connected as shown. Next to each vertex we show the pair of letters stored in it.}
  \end{figure}
\end{example}

\subsection{The rate of interleaved codes}
An immediate observation is as follows.
\begin{proposition}\label{prop:interleave-capacity}
  We have $\cR_{q^k}(\bar{\cC}) = \cR_q(\cC)$, and thus $\ccap_{q^k}(\bar{G}) \ge \ccap_q(G)$.
\end{proposition}
\begin{proof}
    The rate of $\bar{\cC}$ satisfies
    \[
      \cR_{q^k}(\bar{\cC})  = \frac{1}{ns} \log_{q^k} |\cC|^{ks} 
        = \frac{1}{n} \log_{q}|\cC|
        = \cR_q(\cC).
    \]
The second claim is obvious.
\end{proof}

\begin{remark}\label{remark:interleaving capacity}
  Let $\tilde{G} = (\tilde{V}, \tilde{E})$ be a graph such that 
  $\tilde{V} = \{1,\dots,n\} \times \{1,\dots,s\}$ and $((t,\mu), (t',\mu')) \in \bar{E}$ \emph{if and only if} $(t,t') \in E$. Then storage codes on $\tilde{G}$ over $\cQ$ are essentially the same as storage codes on $G$ over $Q^s$, and thus, $\ccap_q(\tilde{G}) = \ccap_{q^s}(G)$. 
  Since $\bar G = \tilde G$ if $k = s$, we focus on the case $k<s$.
\end{remark}
Recall that $\delta(G)$ denotes the size of the largest DAG in $G$.
\begin{proposition}\label{prop:opt-MAIS}
  Let $\cC$ be a storage code on $G$ with $\cR_q(\cC) = \ccap_q(G)$. 
  If $\ccap_q(G) = 1 - \delta(G)/n$, then the interleaved code $\bar{\cC}$ achieves the maximum rate on the corresponding graph $\bar{G}$, and $\ccap_{q^k}(\bar{G}) = \ccap(G)$.
\end{proposition}

\begin{proof}
  We first show that $\delta(\bar{G}) \ge s\delta(G)$. 
  Let $S \subset V$ be the largest DAG set of $G$, and $\bar{S}:=\bigcup_{v \in S}\bar{v}$ (recall that $\bar v=(v,\sigma)$, where $\sigma$ runs over $\{1,\dots,s\}$).
  It is straightforward to check that $\bar{S}$ is also a DAG in $\bar{V}$ of size $s\delta(G)$.
  Indeed, if $(t_1, \mu_1),\dots, (t_m, \mu_m), (t_1, \mu_1)$ is a directed cycle in $\bar{S}$, then $t_1,\dots, t_m, t_1$ is a directed cycle in $S$. 
  Therefore, $\delta(\bar{G}) \ge |\bar{S}| = s|S| = s\delta(G)$. 
  Now by the MAIS bound~\eqref{eq:MAIS},
  \[
    \ccap_{q^k}(\bar{G}) \le 1 - \frac{\delta(\bar{G})}{ns} \le 1 - \frac{|\bar{S}|}{ns} = 1 - \frac{\delta{(G)}}{n} = \ccap_q(G) = \ccap(G).
  \]
  On the other hand, by Proposition~\ref{prop:interleave-capacity}, we have $\ccap_{q^k}(\bar{G}) \ge \ccap_q(G) = \ccap(G)$.
\end{proof}
We remark that given a largest DAG set in $G$, we can easily identify a largest DAG
in $\bar G.$
\begin{corollary}\label{corollary:opt-MAIS}
  Let $G = (V,E)$ be a graph with a storage code $\cC$ that satisfies $\cR_q(\cC) = 1 - \delta(G)/n$. 
  We have $\delta(\bar{G}) = s \delta(\bar{G})$, where $\bar{G}$ is the graph defined by the interleaving procedure.
    In addition, if $S \subset V$ is a largest DAG of $G$, then the set $\bar{S}=\bigcup_{v \in S}\bar{v}$ is a largest DAG in $\bar{G}$.
\end{corollary}
\begin{proof}
  From the proof of Proposition~\ref{prop:opt-MAIS}, we know that $\bar{S}$ is a DAG in $\bar{G}$ and $\delta(\bar{G}) \ge s\delta(G)$.
  On the other hand, by Propositions~\ref{prop:interleave-capacity} and \ref{prop:opt-MAIS}, we have
  \[
    1 - \frac{\delta(G)}{n} = \ccap_q(G) =  \cR_q(C) = \cR_{q^k} (\bar\cC) = \ccap_{q^k}(\bar{G}) \le 1 - \frac{\delta(\bar{G})}{ns},
  \]
  which implies that $\delta(\bar{G}) \le s\delta(G)$. 
  Hence, $\delta(\bar{G}) = s\delta(G)$, and $\bar{S}$ is a largest DAG set.
\end{proof}
In particular, this claim is true if $S$ is an independent set in $G$.
\begin{corollary}\label{prop:opt-IB}
  Let $\cC$ be a storage code on $G$ with $\cR_q(\cC) = \ccap_q(G)$. 
  If $\ccap_q(G) = 1 - \gamma(G)/n$, then the interleaved code $\bar{\cC}$ achieves the maximum rate on the corresponding graph $\bar{G}$, and $\ccap_{q^k}(\bar{G}) = \ccap(G)$.
  Further, if $S \subset G$ is the largest independent set of $G$, then $\bar{S} = \bigcup_{v \in S} \bar{v}$ is the largest independent set of $\bar{G}$.
\end{corollary}

A similar claim holds for the codes that achieve the linear programming bound of
\cite{MazMcgVor2019}.
\begin{proposition} \label{prop:LP-interleaved}
  Let $\cC$ be a storage code on $G = (V,E)$ with $\cR_q(\cC) = \ccap_q(G)$.
  If $\ccap_q(G)$ achieves the LP bound of Theorem~\ref{thm:LP}, then the interleaving code $\bar\cC$ has the maximum possible rate on the corresponding graph $\bar{G} = (\bar{V}, \bar{E})$, and $\ccap_{q^k}(\bar{G}) = \ccap(G)$, where $\bar{G}$ is defined by the interleaving structure.
\end{proposition}
The proof is given in Appendix~\ref{sec:lp-bound}.

Note that generally, for any subgraph $G'\subset \tilde G$ defined on the vertex set $\tilde V,$ there is an obvious relation between the capacities, $\ccap(G') \le \ccap(G).$ Of course, for most subgraphs we do not have a means of constructing good codes; however, for subgraphs of the form $\bar G$ derived from the interleaving construction, we have a way of obtaining optimally sized codes as described above.

\subsection{Some optimal storage codes}
The interleaving construction relies on a good seed code. In this section
we list some known code families that can be used to obtain new optimal storage codes by interleaving, and working out the details of the construction for one of them. Below we use the notation $[n]=\{0,1,\dots,n-1\}$.

\begin{example} \label{ex:lr} 
  Let $l,r$ be positive integers, and $R = \{1 + [r]\} \cup \{-1-[l]\}$, and let $m = \min\{l,r\}$. In Proposition~\ref{prop:1-d-Nd} below we find the capacity of a
  recoverable system $X_R$ defined on the (infinite) graph $\Z_R$.
  Taking a finite subgraph $G = \Z_R \cap [n],$ where $(m+1)|n$,
  we obtain an optimal storage code with $\cR(\cC) = m/(m+1).$
  The code can be obtained from the clique partition construction described above \eqref{eq:clique}, and since 
$\{i(m+1): 0 \le i < n/(m+1)\}$ is the largest DAG set in $G$, we obtain that $\ccap(G) = 1 - \frac{n/(m+1)}{n} = m/(m+1)$. 
\end{example}

Let us apply the interleaving construction for this code. First, 
we color $G$ with colors $\{0, \dots, l+r\}$ such that vertex $t$ is colored 
with color $\tau:=t\,{\text{mod}}\, (l+r+1)$. 
Next, we choose a family of orthogonal partitions $\cM = \{M_0, \dots, M_{r+l}\}$ of shape $(k,s)$ and define $\bar G = (\bar V, \bar E)$ such that $\bar V = \{1,\dots, n\} \times \{1,\dots, s\}$ and $((t,\mu) , (t',\mu')) \in \bar E$ if $(t,t') \in E,$ and column $\mu$ of $M_{\tau}$ intersects column $\mu'$ of $M_{\tau'}$.
Finally, the interleaved code is obtained by the procedure defined before Proposition~\ref{prop:interleaved}, where $\cC$ is the optimal clique partition code described in the previous paragraph.

\begin{example}[Cycles]
Let $C_n$ be an undirected cycle of length $n$. 
If $n$ is even, then the largest independent set of $C_n$ has size $n/2,$ so $\ccap(C_n) \le 1/2$ by ~\eqref{eq:IB}. At the same time, the edge-to-vertex construction yields a storage code $\cC$ of rate $1/2$ over an alphabet $Q$.  If $n$ is odd, then by the LP bound it can be shown that $\ccap(C_n) \le 1/2$. A code of rate $1/2$ can be obtained from the fractional matching construction~\cite{MazMcgVor2019}.
\end{example}

\vspace*{.1in}\begin{example}~ (\cite[Theorem 10]{MazMcgVor2019})
  Let $C_k$ be a cycle of length $k > 3$ and let $B$ be a bipartite graph. Then $\ccap(C_k \boxtimes B) = \frac{1}{2}$, where $\boxtimes$ is the Cartesian product of graphs\footnote{The Cartesian product of two graphs $G_1 = (V_1, E_1)$ and $G_2 = (V_2, E_2)$ is defined by $G = (V,E)$, where $V = V_1 \times V_2$; $((u,u'),(v,v'))\in E$ if and only if $u=v$ and $(u',v') \in E_2$ or $u' = v'$ and $(u,v) \in E_1$.}. The fractional matching construction gives optimal codes.
\end{example}

\vspace*{.1in}\begin{example}~(\cite[Theorem 11]{MazMcgVor2019})
  Let $G$ be a cycle with chords whose endpoints are at least distance $4$ apart on the cycle. We have $\ccap(G) = \frac{1}{2}$, and again, the fractional matching construction gives optimal codes.
\end{example}

In the last two examples, optimality is proved by relying on the linear programming bound.

\subsection{Codes with partial recovery} We say that a storage code $\cC$ over alphabet $Q^k$ has $m$ levels of recovery if for every vertex $v$, by visiting a $1/m$ proportion vertices in $\cN(v)$, we can recover a $1/m$ fraction of the symbols in $v$. Sometimes it may be 
desirable to have codes with multiple levels of recovery, since we do not need to visit all neighbors to recover only a part of the node. 
Consider the interleaved code, for any vertex $(t, \mu) \in \bar{V}$, let $(x_t^{\lambda_1}, \dots, x_t^{\lambda_k})^T$ be the vector stored on it. Suppose that, to 
recover a subset of the coordinates of this vector, it suffices to download information
from a part of the neighbors of $(t, \mu)$, and denote by $\cN(x_t^{\lambda_j})$ the neighbors that determine the values $x^{\lambda_j}_t,j=1,\dots,k$. 
The interleaved construction relies on families of orthogonal partitions, and this
ensures that the sets  $\cN(x_t^{\lambda_j}), j =1,\dots, k$ are disjoint and are of the same size.
Therefore, our interleaving structure yields a storage code with $k$ levels of recovery, and the rate is optimal if the seed code is optimal with respect to the MAIS bound or the LP bound.

\subsection{High-rate storage codes on triangle-free graphs}
It is easy to construct high-rate storage codes on graphs with many cliques, but constructing codes of rate $> 1/2$ on triangle-free graphs is a more difficult problem. 
Recently, \cite{BargZemor2021} constructed infinite families of binary linear storage codes of rate approaching $3/4$.
Following up on this research, two concurrent papers, \cite{barg2} and \cite{HuangXiang2023}, constructed families of binary linear storage codes of rate asymptotically approaching one.
These constructions are based on Cayley graphs obtained as coset graphs of binary linear codes, and they are restricted to the binary alphabet. Here we wish to remark that the code families can be extended to alphabets of size $2^k$ by interleaving.

Let $\cC$ be a binary linear storage code on a graph $G$, and let $\bar \cC$ be the interleaved code on the corresponding graph $\bar G$.
It is easy to see that if $G$ is triangle-free then so is $\bar G,$ and by Proposition~\ref{prop:interleave-capacity}, we have $\cR_{2^k}(\bar{\cC}) = \cR_2(\cC)$. Therefore, interleaving the codes from \cite{barg2} or \cite{HuangXiang2023} we obtain new families of storage codes with asymptotically unit rate over the alphabet $\Z_2^k$.

\section{Recoverable Systems}
\subsection{Recoverable systems: Definitions}
The following definition formalizes the concept of storage codes discussed above
for the case of $V=\Z^d, d\ge 1$. 

\begin{definition}\label{def:main} {\sc(Recoverable systems}) Let $R\subset\Z\backslash\{0\}$ be a finite subset of integers ordered in a 
natural way. An {\em $R$-recoverable system} $X=X_R$ on $\Z$ is a set of bi-infinite sequences
$x\in \cQ^\Z$ such that for any $i\in\Z$ there is a function $f_i:\cQ^{|R|}\to\cQ$ such that  
    $
    x_i=f_i(x_{i+R})
    $
for any $x\in X.$  

More generally, let $R\subset\Z^d\backslash\{0\}$ be a finite subset together with some ordering. 
An {\em $R$-recoverable system} $X=X_R$ on $\Z^d$ is a set of letter assignments ($d$-dimensional words) 
$x$ such that for any vertex $v\in\Z^d$ there is a function $f_v:\cQ^{|R|}\to\cQ$ such that     
$
x_v=f_v(x_{v+R})
$
for any $x\in X.$ Here $v+R=\{v+z : z\in R\}$ is a translation of $v$ by the recovery neighborhood (set) $R$.
\end{definition}
This definition is close to the definition of storage codes \cite{Mazumdar2015} which also allows the dependency of the recovery function on $v\in V$ and also (implicitly) assumes an ordering of the vertices in $V$. Here we wish to note an important point: by defining the recovery set $R$ we effectively introduce an edge (possibly, a directed one) between $v$ and every vertex in $v+R$. Thus, we speak of 
$R$-recoverable systems on $\Z^d$ with edges defined by $R$, and denote the graphs that emerge in this way by $\Z^d_R$. If $d=1$, we omit it from the notation. 

\begin{example}\label{example:lr} Take $V=\Z$ and $R=\{-1,1\}$. Assume that $x_{2i+1}=x_{2i}$ for all $i$. Then we have $f_i(x_{i-1},x_{i+1})=x_{i+(-1)^i}$, that is, $f_i(x_{i-1},x_{i+1})=x_{i-1}$ or $x_{i+1}$ depending on whether $i$ is odd or even.
\end{example}

\begin{definition} {\sc (Capacity, one-dimensional)} Let $X=X_R$ be a one-dimensional $R$-recoverable system. For $n\geq 0$, denote by $B_n(X)$ the restriction of the words in $X$ to the set $[n]:=\{0,1,\dots,n-1\}$, i.e., $B_n(X) = \{x_{[n]} : x\in X\}$.
The {\em rate} of $X$ is defined as 
\begin{equation}\label{eq:cap}
\cR(X)=\cR_q(X):=\limsup_{n\to\infty} \frac{1}{n}\log_q |B_n(X)|.
\end{equation}
The dependence on $q$ will be omitted when it plays no role. 
\end{definition}

Given a recovery set $R\subset \Z\backslash\{0\}$, we are interested in its largest attainable rate  
\begin{equation}\label{eq:capR}
\ccap(\Z_R):=\sup_{X_R\subseteq\cQ^{\Z}}\cR(X_R)
\end{equation}
of $R$-recoverable systems on $\Z$, calling it the capacity of the graph $\Z_R$.

\begin{example}
Consider the $R$-recoverable system from Example~\ref{example:lr}. Since every odd symbol is uniquely determined by the two adjacent symbols in even positions,
we have $B_n(X)\leq q^{\lceil n/2\rceil},$ and since every symbol can occur in the even positions, we have $B_n(X)\geq q^{\lfloor n/2\rfloor}$. 
Thus, $\cR(X)=\frac{1}{2}$.
\end{example}

Similarly, we define the capacity of multidimensional systems.
\begin{definition}\label{def:CapD} {\sc (Capacity, $d$ dimensions)} 
Let $R\subset\Z^d\setminus \{0\}$ be a finite set (the recovery set) and let $\cB_{[n]^d}(X)=[n]^d$ be the $d$-dimensional cube.
Given a recoverable system $X=X_R$, define its rate as
  $$
  \cR_q(X)= \limsupup{n} \frac{1}{n^d}\log_q |\cB_{[n]^d}(X)|,
  $$
and let $\ccap(\Z^d_R)=\sup_{X_R\subseteq Q^{\Z^d}} \cR(X_R)$ be the largest rate associated with the set $R$.
\end{definition}

At this point one may wonder if the capacity value depends on our choice of cubes in this definition. For the invariant case we prove that the answer is no, although a rigorous treatment is somewhat technical, and we postpone it till Sec.\ref{sec:Zd}.

\subsection{Recoverable systems on \texorpdfstring{$\Z$}{}}
Bounds and constructions for finite graphs can be extended to infinite graphs such as $\Z^d$. 
In this section, we state a few easy results that exemplify the above methods for the case of $d=1$. Analogous claims in higher dimensions are discussed in the next subsection.

Let $R\subset \Z\setminus\{0\}$ be a finite set and consider $R$-recoverable systems on the infinite graph $\Z_R$. For all $n>0$ let us consider the subgraph $\Z_{R,n}:=\Z_R\cap[n]$. 
To bound above the capacity $\ccap(\Z_R)$ we first find $\ccap(\Z_{R,n})$ and argue that these quantities are related.
This follows from the observation that we can disregard the boundary effects because the proportion of points whose recovery regions do not fit in $[n]$ vanishes as $n$ increases. Thus, we can place constant values on these
points with a negligible effect on the rate of the code.
\begin{lemma}\label{lem:fin_2_inf}
Let $R\subset \Z\setminus\{0\}$ be a finite set and $(c_n)_{n=1}^{\infty}$ an infinite sequence of real numbers such that $\ccap(\Z_{R,n})\leq c_n$ for all $n\ge1$. Then, $\ccap(\Z_R)\leq \limsup_{n\to\infty} c_n$. 
\end{lemma}
\begin{proof} 
For a given $n\ge 1$, let $A_n$ be the set of vertices in $\Z_{R,n}$ in which their neighborhood in the graph $\Z_R$ does not fully belong to the graph $\Z_{R,n}$, i.e., $A_n =\{m \in [n] : m+R \nsubseteq [n]\}$. 
Since $|R|$ is a constant independent of $n$, starting from some value of $n$ we have $|A_n|\le a$ for some absolute constant $a$.

Let $X$ be an $R$-recoverable system and recall that $B_n(X)$ is its restriction to $[n]$. For a given vector $z = (x_i)_{i\in A_n}\in \cQ^{|A_n|}$, let $B_n^{z}(X)$ be the set of words in $B_n(X)$ which match the values of $z$ over the positions in $A_n$. The code $B_n^{z}(X)$ is a storage code over the graph $\Z_{R,n}$ and hence $$n^{-1} \cdot \log_q|B_n^{z}(X)| \leq \ccap(\Z_{R,n})\leq c_n.$$ Since the inequality holds for all $z \in \cQ^{|A_n|}$ and $|A_n|\leq a$, we conclude that $$n^{-1} \cdot\log_q|B_n(X)| \leq c_n+a/n,$$
which verifies the statement of the lemma.
\end{proof} 

According to Lemma~\ref{lem:fin_2_inf}, it is possible to calculate the capacity of several recovery sets $R$. 
\begin{proposition}\label{prop:1-d-Nd}
Let $l,r$ be positive integers, let $R=\{1+[r]\}\cup\{-1-[l]\},$ and let $m=\min\mathset{l,r}$. Then,  
    $$
    \ccap(\Z_R)=\frac{m}{m+1}.
    $$
\end{proposition}
\begin{proof} 
Assume without loss of generality that $r\leq l$, so $m=r$, and let $\Z_R$ be the graph describing the $R$-recoverable system. The upper bound follows by noting that for all $n>1$ the set of vertices $\{k(r+1):k\in\Z\}\cap[n]$ is a DAG set in the graph $\Z_{R,n}$ and thus 
\begin{align*}
\ccap(\Z_{R,n}) & \leq 1 - \frac{\delta(\Z_{R,n})}{n} \leq 1 - \frac{\lceil \frac{n}{r+1}\rceil}{n}\leq 
\frac r{r+1}.
\end{align*}
Together with Lemma~\ref{lem:fin_2_inf} this implies that $\ccap(\Z_R)\leq r/(r+1)$.

To show the reverse inequality we use the clique partition construction (\ref{eq:clique}). Start with partitioning $\Z$ into segments of length $r+1$, setting $\Z=\bigcup_{k\in\Z}\,k\cdot[r+1].$  We aim to construct a code in which every vertex in the segment can be recovered from the other vertices in it. In other words, $\Z_R$ is a disjoint union of $(r+1)$-cliques, and the construction described above before \eqref{eq:clique} yields a code of rate $r/(r+1).$
\end{proof}

In the case of {\em shift invariant} recoverable systems the upper bound of this proposition
was derived in \cite{ElishcoBarg2020}; however, the authors of  \cite{ElishcoBarg2020} 
stopped short of finding a matching construction under this assumption. The challenge in the shift invariant case is to find a recoverable system with the \emph{same} recovery function for all symbols. 

The next claim concerns recoverable systems in which the recovery functions use only two symbols for the recovery process, but those symbols are not 
necessarily adjacent to the symbol to be recovered.
\begin{proposition}\label{prop:capR}
Let $l,r$ be positive integers and let $R=\mathset{-l,r}$. Then for any $q\ge 2$ 
\[\ccap(\Z_R)=\frac{\gcd(l,r)}{l+r}.\]
\end{proposition}
\begin{proof}[Proof] 
We use a similar technique to the one used in the previous proof. First, we present an upper bound on the capacity. Denote by $d:=\gcd (l,r)$, $m=l+r$, and assume without loss of generality that $l\leq r$. For a sequence $x\in \cB_n(X)$ for $n$ large enough, we show how we can leave only the symbols in positions $i\in A:=\mathset{km+j ~:~ k\in \N,\; j\in [d]}$ since all the other symbols can be uniquely determined by the symbols in the positions of $A$. 
Notice that $x_l,\dots,x_{l+d-1}$ can be recovered at first since $x_{[d]},x_{m+[d]}\in A$. As a matter of fact, this can be done for every $m$-block in $x$, 
i.e., $x_{km+l+[d]}$ can be recovered for every $k\in [n/m]$. 
At the next step, it is possible to recover all the symbols $x_{2l+[d]}$ since during the previous 
step we recovered the symbols in positions $x_{l+[d]}$ and in $x_{2l+r+[d]}$. Again, this can be done in every $m$-block so it is possible to recover 
the symbols $x_i$ for $i\in\mathset{km+2l+[d] ~:~ k\in\N}$. We continue in the same way such that at the $t$-th step, we recover the set $x_{km+tl+[d]}$ for $k\in\N$. We are only left to show that the symbols in $x_d,x_{d+1},\dots, x_{l-1}$ are recovered. Since $d=\gcd(l,r)$, there exists $t_1$ such that $t_1 l=k_1 m+d$ for some $k_1$ which implies that $x_{km+d+[l]}$ are recovered. 

The lower bound is given by the following encoding process. 
The information symbols are stored in positions $i\in\mathset{km+[d] ~:~ k\in \N}$. Then the parity symbols are added in the same order as above. 
\end{proof}

\begin{example}
We demonstrate the encoding process for a recoverable system $X_R$ with $R=\mathset{-l,r}$ over a three-letter alphabet. Take $l=6, r=4$, then $\gcd(4,6)=2$.
We show the encoding process for positions $10,11,\dots,19,$ and note that similar steps are made 
for every other $(l+r)$-block. First, we put data symbols in positions $10k,11k$ for all $k\in \Z.$
The other symbols are chosen so as to satisfy the parity $x_i+x_{i-l}+x_{i+r}=0.$
To illustrate this, find $x_{14}=-(x_{10}+x_{20})$ and $x_{15}=-(x_{11}+x_{21}),$ then find
$(x_{18},x_{19})$ from $(x_{14},x_{15};x_{24},x_{25}),$ etc.
\end{example}

\begin{example}
For the same set $R=\{-6,4\}$, put a codeword of the length-5 repetition code on each of the mutually disjoint sets of consecutive symbols in the even positions and in the odd positions. For instance, we can take those sets to be $\{10k+2j\}$ and $\{10k+2j+1\}$ with $0\leq j \leq 4, k\in\Z$. Then clearly for every $i\in\Z$, the symbol $x_i$ is a part of the codeword of the repetition code that includes either position $i-6$ or position $i+4,$ and we can recover $x_i$ by accessing one of those positions as appropriate.
\end{example}

Let $B\subset \Z_+$ be a finite or infinite set. We call a set $A\in \Z$ \emph{$B$-avoiding} if its difference set $|A-A|$ is disjoint from $B$, i.e., $|b_1-b_2|\notin B$ for all $b_1,b_2\in A$. Let $a_n$ be the size of the largest set $A_n\subseteq [n]$ that is $B$-avoiding and define $\cR(B) = \limsup_{n\to\infty}a_n/n$. 
This quantity has been extensively studied in the literature, initially with $B$ being the set of all whole squares~\cite{S78I} and later values of other polynomials; see, e.g., \cite{Rice2019}.

\begin{proposition}\label{prop:relative}
Let $0<r_1<r_2<\cdots<r_s$ and $0<l_1<l_2<\cdots <l_t$ be positive integers and $R = \{-l_t,\ldots,-l_2,-l_1,r_1,r_2,\ldots,r_s\}$. Then, $$\ccap(\Z_R)\leq 1- \max\{\cR(\{l_1,\ldots,l_t\}),\cR(\{r_1,\ldots,r_s\})\}.$$\end{proposition}
\begin{proof}
Assume without loss of generality that $\cR(\{l_1,\ldots,l_t\}) \leq \cR(\{r_1,\ldots,r_s\})$ and let $A_n$ be the largest subset of $[n]$ that is 
$\{r_1,r_2,\ldots,r_s\}$-avoiding. Then, $A_n$ is a DAG set in the graph $\Z_{R,n}$ and thus $\ccap(\Z_{R,n})\leq 1 - |A_n|/n$. 
Indeed, all edges in $A_n$ are going left.
The statement follows from Lemma~\ref{lem:fin_2_inf}. 
\end{proof}
Consider for example the set $B_m=\{1,2,4,\ldots,2^m\}$ for $m\geq 1$. Then, $\cR(B_m) = 1/3$ (achieved by the $B_m$-avoiding set $\{3j : j\geq 0\}\cap [n]$ for all $n$). Then, according to Proposition~\ref{prop:relative}, it is possible to derive that for all $t,s\geq 1$, and $R=\{-2^t,\ldots,-4,-2,-1,1,2,4,\allowbreak\ldots,2^s\}$, $\ccap(\Z_R) \leq 1-1/3$. Equality holds in this case since by Proposition~\ref{prop:1-d-Nd}, $\ccap(\Z_R) \geq \ccap(\{-2,-1,1,2\}) = 2/3$. 

\subsubsection{Interleaving and recoverable systems}
The interleaving construction can also be applied to recoverable systems, and sometimes it yields optimal-rate storage codes on infinite graphs. For instance, let $X_R$ be a one-dimensional recoverable system defined in Proposition~\ref{prop:1-d-Nd}, and let $\Z_R$ be its corresponding (infinite) graph. As in Example~\ref{ex:lr}, we first color $\Z_R$ with $\{0, \dots, (r + l)\}$ such that vertex $t\in \Z_R$ is colored with $t\,\text{mod\,}(r + l + 1)$, and then we find a family of orthogonal partitions $\cM = \{M_0, \dots, M_{l+r}\}$ of shape $(k,s)$. 
Now define the infinite graph $\bar \Z_R$ on the vertices $\Z \times \{1,\dots, s\}$ such that there is an edge from $(t,\mu)$ to $(t',\mu')$ if and only if

\begin{itemize}\item there is an edge from $t$ to $t'$ in $\Z_R$, and 
\item column $\mu$ of $M_{t\,\text{mod\,}(r + l + 1)}$ intersects column $\mu'$ of $M_{t' \,\text{mod\,}(r + l + 1)}$.
\end{itemize}
The interleaved codewords are obtained by placing $(x_t^{\lambda_1}, \dots, x_t^{\lambda_k})^T$ on vertex $(t,\mu)$, where $\lambda_i = M_{c(t)}(i,\mu)$, $i = 1, \dots, k.$ Letting $x^1, \dots, x^{ks} \in X_R$
run over all the possible choices in $X_R,$ we 
obtain the interleaved code $\bar X_R$.
In general, the interleaving construction works for every recoverable system $X_R$, and yields recoverable systems $\bar X_R$ on an infinite graph defined by the interleaving.
To define the rate, 
let $Y$ be a recoverable system defined on $\bar Z_R$ and set
    $$
    \cR_q (Y):= \lim_{n \to \infty} \frac{1}{ns} \log_q B_{[n] \times \{1,\dots, s\}} (Y),
    $$ 
and finally let $\ccap_q(\bar \Z_R)=\sup_Y\cR_q(Y)$.
We observe that  Lemma~\ref{lem:fin_2_inf} can be generalized to bound $\ccap_{q^k}(\bar \Z_R)$, namely, $\ccap_{q^k}(\bar \Z_R) \le \limsup_{n \to \infty} \ccap(\bar \Z_R \cap [n] \times \{1,\dots,s\})$ \footnote{The proof follows the proof of Lemma~\ref{lem:fin_2_inf}, with the only difference that we take $A_n = \{ (m,\mu) \in [n] \times \{1,\dots, s\} : m + R \notin [n]\}$. }.
Therefore, if  $\ccap_q(X_R)$ is bounded above by the DAG bound (as it happens, for instance, for recoverable systems defined in Propositions~\ref{prop:1-d-Nd},~\ref{prop:capR}, and \ref{prop:relative}), then the interleaving construction yields a capacity-achieving recoverable system $\bar X_R$ on $\bar \Z_R$.

\subsection{Recoverable systems on \texorpdfstring{$\Z^d$}{}}\label{sec:Zd}
In this section, we present results concerning the capacity of multidimensional systems and their relationship with one-dimensional capacity. 
As in the one-dimensional case, we derive results for the $d$-dimensional grid relying on capacity estimates for storage codes in its finite subgraphs, arguing that the boundary effects can be disregarded in the limit of large $n$. Let $\Z_{R,n}^d:=\Z_{R}^d\cap[n]^d$ and observe that
Lemma~\ref{lem:fin_2_inf} affords the following generalization.
\begin{lemma}\label{lem:fin_d_inf}
  Let $R \subset \Z^d \setminus \{0^d\}$ be a finite set and let $(c_n)_{n}$ be a sequence of real numbers such that $\ccap(\Z^d_{R,n}) \le c_n$ for all $n \ge 1$.
  Then, $\ccap(\Z^d_{R}) \le \limsup_{n \to \infty} c_n$. 
\end{lemma}
\begin{proof} 
Again let $A_n^d:=\{m\in [n]^d: m+R\subsetneq [n]^d\}.$ Since $R$ is a finite region whose size is independent of $n$, and since the number of $(d-1)$-dimensional faces of the cube $[n]^d$ is $2d$, we have $|A_n^d|=O_d(n^{d-1})$. 

Arguing as in Lemma~\ref{lem:fin_2_inf}, let $X$ be an $R$-recoverable system, let $z\in \cQ^{|A_n^d|}$ be a fixed word, and let $B_{[n]^d}^z$ be the restriction of $X$ to $[n]^d$ that matches $z$ on $A_n^d.$
The set $B_{[n]^d}^z$ forms a storage code for $\Z_{R,n}^d,$ and this holds for every choice of $z.$
Since there are $q^{|A_n^d|}=O(q^{n^{d-1}})$ such choices, using the assumption of the lemma, we obtain
   \begin{align*}
   \frac1{n^d}\log_q (O(q^{n^{d-1}})|B_{[n]^d}^z|)\le O\Big(\frac 1n\Big)+c_n.
   \end{align*}
Since this is true for every system $X$, in the limit of $n\to\infty$ we obtain the claim of the lemma.
\end{proof} 

This lemma is used in Section~\ref{sec:ell-metrics} below to derive some capacity results for $\Z^2.$
In this section we show that in some cases capacity of graphs over $\Z^d$ is tightly related to 
capacity of one-dimensional graphs. For the remainder of this section, we assume 
an additional property, which we call invariance. 
Recall that if $R$ is a recovery region for $\Z^d$, then it defines a 
modified graph $\Z^d_R$ obtained by adding edges from every vertex $v\in\Z^d$ to the vertices in its recovery region. We say that $R$ is invariant if for any DAG $S$ in $\Z^d_R,$ the shifted set
$a+S:=\{a+s,s\in S\}$ also forms a DAG set. By extension, an $R$-recoverable system $X$ on 
$\Z^d$ is called invariant if so is $R$. 
While in general the recovery function depends on the vertex to be recovered, in shift invariant systems, it does not depend on the vertex, and has a fixed shape up to translation.

We begin by demonstrating that there is no loss of generality in limiting oneself to $n$-cubes in the capacity definition, Definition~\ref{def:CapD}. Consider an $R$-recoverable system $X$ with a finite recovery region $R\subseteq \Z^d\setminus \{0\}, d\ge 2$. 

For any set $S\subseteq \Z^d$ denote by $\cB_S(X)$ the restriction of the words in $X$ to the set $S.$ 
Now assume $X$ is an invariant system, let $S_1,S_2\subset \Z^d$, and notice that $|\cB_{S_1\cup S_2}(X)|\leq |\cB_{S_1}(X)|\cdot |\cB_{S_2}(X)|,$ or
   $$
   \log_q |\cB_{S_1\cup S_2}(X)|\leq \log_q |\cB_{S_1}(X)| + \log_q |\cB_{S_2}(X)|.
   $$
   
Moreover, by the definition of an $R$-recoverable system, for any finite set $S\subseteq \Z^d$ and any $a\in \Z^d$ we have $|\cB_S(X)|=|\cB_{a+S}(X)|$. 
Our arguments will use the Ornstein-Weiss Theorem (see \cite{ornstein1987entropy} and a detailed proof in \cite{Kri2010}), which we cite in the form adjusted to our needs.
   \begin{theorem} {\rm [Ornstein-Weiss]}
\label{th:or_we}
Let $f$ be a function from the set of finite subsets of $\Z^d$ to $\R$, satisfying the following: 
\begin{itemize}
    \item $f$ is sub-additive, i.e., for every finite subsets $S_1,S_2$, $f(S_1\cup S_2)\leq f(S_1)+f(S_2)$. 
    \item  $f$ is translation invariant, i.e., for any $a\in \Z^d$ and $S\subseteq \Z^d$, $f(a+S)=f(S)$.
\end{itemize} 
Then for every sequence $(S_i)_i\subseteq \Z^d$ such that $\lim_{i\to\infty} \frac{|(a+S_i)\triangle S_i|}{|S_i|}=0$ for every $a\in \Z^d$, the limit
$\lim_{i\to\infty} \frac{f(S_i)}{|S_i|}$ exists, is finite, and it does not depend on the choice of the sequence $(S_i)_i$.
   \end{theorem} 

A sequence $(S_i)_i$ for which $\lim_{i\to\infty} \frac{|(a+S_i)\triangle S_i|}{|S_i|}=0$ for every $a\in \Z^d$, is called a F{\o}lner sequence, and it can be thought of as a sequence of shapes whose boundary becomes negligible. Theorem~\ref{th:or_we} implies that for invariant systems, for every F{\o}lner sequence $(S_i)_i$ we have 
  $$
\cR(X)=\lim_{n\to\infty} \frac{\log_q |\cB_{[n]^d}(X)|}{n^d}= \lim_{i\to\infty} \frac{\log_q |X_{S_i}|}{|S_i|}.
  $$

In particular, since the sequence $([n]^d)_{(n\ge 1)}$ is a F{\o}lner sequence, the rate $\cR(X)$ is well defined. The capacity of a recovery region $R$ (or of the graph $\Z^d_R$) is defined
as the largest attainable rate of $R$-recoverable systems. In this section we limit ourselves to invariant systems, and we denote their capacity by $\ccap_r(R)$ (instead of $\ccap_r(\Z^d_R)$),
where the subscript $r$ refers to invariance. Thus, $\ccap_r(R)=\sup \cR(X_R),$ where the supremum is taken over all invariant systems in $Q^{\Z^d}$. Notice that all the definitions given so far are valid when restricting ourselves to the set of invariant recoverable systems.

We will now establish a connection between the capacity of invariant recoverable systems in one dimension and their higher-dimensional counterparts. 
To do so, we introduce the following definition. Let $\bfn=(n_1,\dots,n_d)\in \N^d$ be a $d$-dimensional vector. 
As above, $\cB_{[\bfn]}(X)$ denotes the set of words in $X$ that are restricted to the coordinates $[\bfn]=[n_1]\times \cdots \times [n_d]$. 
Given a sequence of vectors $(\bfn_i)_i$, we say that $\bfn_i\to \infty$ as $i\to\infty$ if the sequence 
    $$
    \min \bfn_i:=\min \mathset{(n_1)_i,\dots,(n_d)_i}\to\infty
    $$
as $i\to\infty$. 
It is straightforward to show that if a sequence $({\bfn_i})_i\subset \Z^d$ satisfying $\lim_{i\to\infty}\bfn_i=\infty$, then $\lim_{i\to\infty} \frac{|(a+[\bfn_i])\triangle [\bfn_i]|}{|[\bfn_i]|}=0$ for every $a\in \Z^d$, i.e., $({\bfn_i})_i$ is a F{\o}lner sequence (here $|[\bfn]|=\prod_{j=1}^n n_j$).

Let 
  \begin{equation}\label{eq:Rd}
    R^d = \bigcup_{i=1}^d (\{0\}^{i-1} \times R \times \{0\}^{d-i})
  \end{equation}
(below we call such regions {\em axial products}; see Thm.~\ref{th:ax}). It is intuitively clear that $\ccap_r(R)\leq \ccap_r(R^d)$.
Below we formally prove 
this claim together with an upper bound on $\ccap_r(R^d)$. We will use a multivariate generalization of the well-known Fekete lemma, which extends the original claim to subadditive functions on $d$-dimensional lattices. A general version of this statement appears, for instance, in \cite[Lemma 15.11]{Georgii2011}, although apparently it has been known for a long while. We will use a version adjusted to our needs, considering sequences indexed by integer vectors and subadditive along each of the coordinates.

For a vector $\bfn=(n_1,\dots, n_d)\in \N^d$ and $c\in \N$, let 
$\bfn^{(i,+c)}$ be the vector obtained from $\bfn$ by replacing $n_i$ with $n_i+c,$ and let $\bfn^{(i,c)}$
be obtained upon replacing $n_i$ with $c.$

\vspace*{.1in}
\begin{theorem}\cite[Th. 1]{cap2008},~\cite{cap2022}\label{thm:Fekete} Let $\xi:\N^d\to [0,\infty)$ be a sequence 
that satisfies   
   $$
  \xi(\bfn^{(i,+c_i)})\le \xi(\bfn)+\xi(\bfn^{(i,c_i)})
   $$
 for all $c_i\in \N, i=1,\dots,d$. Then, the limit $\lim_{\bfn\to\infty}\frac{\xi(\bfn)}{\prod_j n_j}$ exists and is equal to $\inf_{S\subseteq \N^d}\frac{\xi(S)}{|S|}$.
\end{theorem}
\vspace*{.1in}

\begin{proposition}\label{prop:invariant-d} 
  Let $R \subset \Z \setminus \{0\}$ be a finite set and for every $n$, let $S_n$ be an invariant DAG set in $\Z_{R,n}$. 
  Then $$\ccap_r(R)\leq \ccap_r(R^d) \leq \limsup_{n\to\infty} (1-|S_n|/n).$$
\end{proposition}
\begin{proof} It will suffice to limit ourselves to considering restrictions of $X_{R,n}$ to cubes, so below $\cB_{[\bfn]}$ denotes a cube $[n]^d.$
The first inequality is obvious by stacking one-dimensional words from $X_{R,n}$ to form a word in $X_{R^d}$. Namely, fix any $(i_1,\dots,i_{d-1})\in[n]^{d-1}$ and place a word of $X_{R,n}$ in the $n$ positions given by varying the last coordinate $i_d$ inside $\cB_{[\bfn]}$; repeat for all the choices of $(i_1,\dots,i_{d-1}).$ Since every symbol can be recovered from its one-dimensional neighborhood along the varying coordinate, it can also be recovered from its $R^d$-neighborhood. 
  
Let us prove the second inequality. 
For $a\in\N$ we have 
  $$
  |\cB_{[\bfn^{(i,+a)}]}(X_{R^d})|\leq |\cB_{[\bfn]}(X_{R^d})|\cdot |\cB_{[\bfn^{(i,a)}]}(X_{R^d})|,
  $$
or
   $$
\log |\cB_{[\bfn^{(i,+a)}]}(X_{R^d})|\leq \log |\cB_{[\bfn]}(X_{R^d})|+ \log |\cB_{[\bfn^{(i,a)}]}(X_{R^d})|.
   $$
Thus, the log-volume sequence is subadditive coordinate-wise, so Theorem~\ref{thm:Fekete} applies, 
yielding that the limit in \eqref{eq:cap} exists, i.e., 
  $$
  \cR(X_R)=\lim_{n\to\infty} \frac{\log |\cB_{[\bfn]}(X_{R^d})|}{n^d}.
  $$
Moreover, this limit equals the infimum on the choice of the shape, so taking $\bfn_i=(i,0,0,\dots,0)$ for the upper bound on capacity we obtain that 
  $$
  \cR(X_{R^d})\leq \frac{\log \cB_{[\bfn_i]}(X_{R^d})}{i}.
  $$
Since $\cB_{[\bfn_i]}(X_{R^d})$ is the restriction of $X$ to $[\bfn_i]$, and since $R^d$ is constructed as an axial product, we obtain that $S_n$ is a DAG set in $\Z_{R^d,n}$. Thus, $\frac{\log \cB_{[\bfn_i]}(X_{R^d})}{i}\leq 1- |S_n|/n$ which implies that 
  $$
  \ccap_r(R^d)\leq \limsup_{n\to\infty} \Big(1-\frac{|S_n|}{n}\Big). 
  $$
\end{proof}

As an immediate corollary of this proposition, we obtain the following. 
\begin{corollary}
Let $R\subseteq \Z\setminus \{0\}$ be a finite set and let $X$ be an invariant $R$-recoverable system that attains 
the MAIS bound \eqref{eq:MAIS}. 
Let $R^d$ be given by \eqref{eq:Rd}, then $\ccap(R^d)=\ccap(R)$.
\end{corollary}

\section{Capacity of \texorpdfstring{$\Z^2$}{} with \texorpdfstring{$l_1$}{} and 
\texorpdfstring{$l_\infty$}{} recovery regions}\label{sec:ell-metrics}

The results in the previous sections enable us to study the capacity of graphs that represent a storage network over the two-dimensional grid. We start with recovery sets formed by radius-$r$ balls under the $l_1$ and $l_{\infty}$ metrics. 
For $v=(i_1,j_1),u=(i_2,j_2)\in\Z^2$, 
\begin{align*}
& d_1\big( (i_1,j_1),(i_2,j_2) \big) = |i_1-i_2| + |j_1-j_2|, & \\
& d_\infty\big( (i_1,j_1),(i_2,j_2) \big) = \max\{|i_1-i_2|,|j_1-j_2|\}. &
\end{align*}
The $l_1$ metric on $\Z^2$ is sometimes called the Manhattan distance (or even Lee distance \cite{E11}, although this usage is not fully accurate).  A sphere in this metric is a rhombic pattern whose exact shape depends on the parity of the radius (see below).
A sphere in the metric $d_\infty$ is a $(2r+1)\times (2r+1)$ square. In both cases to argue about recovery, we add to $\Z^2$ the edges that connect every vertex with its neighbors in the sphere of radius $r$ about it, and denote the resulting graph by $G_r^\alpha$, for $\alpha=1,\infty$. To find the capacity of these graphs, we rely on Theorem \ref{cor:covering} applied to their finite subgraphs. These subgraphs do not have transitive automorphisms, but finding the capacity of the system essentially amounts to constructing a perfect covering of the graph by anticodes. This claim is proved in the following general statement.

\begin{proposition} Let $D_G(r)$ be the largest size of an anticode of diameter $r$ in $\Z^2$ in the $d_\alpha$ metric, $\alpha=1,\infty,$ and suppose that $\Z^2$ admits a tiling with anticodes of size $D_G(r)$. Then 
   $$
   \ccap(G_r^\alpha)=1-\frac1{D_G(r)}.
   $$
\end{proposition}
\begin{proof}
Suppose that there is a perfect covering of $\Z^2$ with largest-size anticodes of diameter $r$. Then the graph $G:=\Z^2_n$ contains a perfect covering, except for the subset formed of the points at distance $r$ or less from one of the boundaries. As a result, $\ccap(G_r)\le 1-\frac 1{D_G(r)}$, and also
$\ccap((\Z^2)_r)\le 1-\frac 1{D_G(r)}$ on account of Lemma~\ref{lem:fin_d_inf}. At the same time, since the anticodes provide a perfect covering
of $\Z^2$, we have that $\ccap(G_r)\ge 1-\frac 1{D_G(r)}$.
\end{proof}

This enables us to find the capacity values of the graphs $G^\alpha_r$.
\begin{theorem}\label{thm:1-infty}
For all $r\geq 1$ it holds that
\begin{align*}
& \ccap(G^1_r) = 1 - \frac{1}{D_1(r)}, \\
&\ccap(G^\infty_r) = 1 - \frac{1}{(r+1)^2}, 
\end{align*}
where $D_1(r)$ is the size of a maximum anticode of diameter $r$ in the $l_1$ sphere, and
\begin{equation*}
  D_1(r) = 
  \begin{cases}
      \frac{(r+1)^2}{2} & r \text{ odd},  \\
      \frac{r^2}{2} + r + 1 & r \text{ even}.
  \end{cases}
\end{equation*}
\end{theorem}
\begin{proof}
For the graph $G^\infty_r$, anticodes of diameter $r$ are simply squares of size $(r+1)\times (r+1),$ so $D_{G^\infty_r}(r) = (r+1)^2$. The squares tile the graph, giving a perfect covering, and thus
    $$ 
    \ccap(G^\infty_r) = 1 - \frac{1}{{(r+1)^2}}.
    $$
The result for the $G^1_r$ metric is obtained from perfect tiling which exists for maximal anticodes in the $l_1$ metric (see e.g.~\cite{E11}).
In detail, for $r = 2k$, the largest anticode with diameter $r$ is an $l_1$ ball of radius $k$, which forms a tiling in $\Z^2$. Its size is easily found by induction to be $r^2/2 + r +1.$
For $r = 2k+1$, the largest anticode $\cD_k$ can be defined recursively as follows. Let $\cD_0$ be a shape formed by two adjacent (with $l_1$ distance one) points of $\Z^2$, and let $\cD_m, m=1,\dots,k$ be formed of $\cD_{m-1}$ and all the points that are adjacent to at least one of points in $\cD_{m-1}$. Note that $\cD_k$ also generates a tiling, and $|\cD_k| = (r+1)^2/2$.
\end{proof}
\begin{remark}
Note that this theorem relies on a stronger assumption that Theorem~\ref{cor:covering}, namely that the there is a tiling of the graph with translations of a largest-size anticode. The existence of tilings and more generally, diameter perfect codes in $\Z^d$, is an active research topic, see~\cite{zhang2022linear,zhang2019nonexistence} for recent additions to the literature.
\end{remark}

Now suppose that the recovery set is not symmetric, for instance, a direct product of two non-symmetric segments. We have the following result.
\begin{theorem}
Let $l,r,b,a$ be positive integers and consider the recovery set $R = [-l,r]\times [-b,a]$ 
{where $0\le r< l$ and $0\le b< a$}. 
Let $X_R$ be a two-dimensional $R$-recoverable system on $\Z^2$.
Then 
\begin{equation}\label{eq:nsym}
\ccap(G_R)=1-\frac{1}{{(r+1)(b+1)}}.
\end{equation}
\end{theorem}

\begin{proof}
To show that $\ccap(G_R)$ is less than the right-hand side of \eqref{eq:nsym} 
we observe that $$
S= \{{\left(i(r+1),j(b+1)\right)}: i,j\geq 0\}
$$ 
is a DAG set in the graph $G_R$.
Indeed, there is an edge from $v_1 = \left(i_1(r+1),j_1(b+1) \right)$ to $v_2 =  \left(i_2(r+1),j_2(b+1) \right)$ only if $i_1 \le i_2$ and $j_1 \le j_2$. 
Through this, we can define a topological order of all vertices in $S$ and the claim follows from
Theorem~\ref{th:DAG_set}.

For the lower bound, let 
$$
S_1=\mathset{(i,j)\in\Z^2 ~:~ i\in [r], j\in -[b]}.
$$
For $x,y\geq 0$, the subset of nodes 
$({x(r+1),y(b+1)}) + S_1$ 
forms a clique. 
Since $|S| = |S_1|$, the proof is concluded by \eqref{eq:clique}.
\end{proof}

As already noted \eqref{eq:Rd}, one way to construct two-dimensional systems is to form an {axial product} of two one-dimensional systems. As before, for positive integers $l,r,b,a$ let 
{$R_1=[-l,r]$ and $R_2=[-b,a]$}. Fix a finite alphabet $\cQ$ and consider one-dimensional recoverable systems $Y_1=Y_{R_1}$ and $Y_2=Y_{R_2}$. The {\em axial product} of $Y_1$ and $Y_2$ is a two-dimensional system $X=Y_1\times Y_2$ over $\cQ$ with the recovery set 
$(R_1\times\{0\}) \cup (\{0\}\times R_2)$. In words, in the system $X$ the symbol $x_i$ is a function of the $l$ symbols to its left, $r$ symbols to its right, $a$ symbols above, and $b$ symbols below it.
The axial product construction is different from the direct product $[-l,r]\times[-b,a]$ in that it results in a cross-shaped rather than a rectangular recovery region.

\begin{theorem}
\label{th:ax}
For given positive integers $l,r,b,a$, let $R_1=[-l,r]$ and $R_2=[-b,a],$
and let $X=Y_1\times Y_2$ be the axial product of the one-dimensional systems $Y_1=Y_{R_1}$ and $Y_2=Y_{R_2}$.
Denote the recovery set $(R_1 \times \{0\}) \cup (\{0\} \times R_2)$ by $R^2$.
Then the capacity is
\[\ccap(\Z^2_{R^2})=\frac{t}{t+1},\] 
where $t=\max \mathset{\min\mathset{l,r}, \min \mathset{a,b}}$. 
Moreover, a system $X$ that attains this value can be obtained from the one dimensional $R$-recoverable system with $R=\mathset{j:0<|j|<t}$.
\end{theorem}
\begin{proof}
  We first observe that $\ccap(\Z^2_{R^2}) \ge \max \mathset{\ccap(\Z_{R_1}), \ccap(\Z_{R_2})}$ by stacking words from $Y_{R_1}$ or $Y_{R_2}$ on top of each other to generate a recoverable system.
  The recovery set is $\{R_1 \times \{0\}\} \subset R^2$ or $\{\{0\} \times R_2\} \subset R^2$.
  Then by Proposition~\ref{prop:1-d-Nd} we have
  $$
    \ccap(\Z^2_{R^2}) \ge \max \left \{ \lim_{n \to \infty} \frac{|\Z_{R_1}| n }{n^2} ,\lim_{n \to \infty} \frac{|\Z_{R_2}| n }{n^2} \right \} = \max \mathset{\ccap(\Z_{R_1}), \ccap(\Z_{R_2})} = \frac{t}{t+1}.
  $$
  To show the other direction, note that the set of vertices 
  $$
    S := [n]^2 \cap \left \{ (k(t+1),0) + (i,i), (0, k(t+1)) + (i,i), k,i \in \Z  \right \} 
  $$
  is a DAG set in $\Z^2_{R^2,n}$ (as before, $\Z^2_{R^2,n} := \Z^2_{R^2} \cap [n]^2$). 
  It remains to calculate the cardinality of $S$.
  Assume without loss of generality that $n = m(t+1)$ for some $m \in \N$, then we have 
      \begin{align*}
      |S| &= 2(1 + (t+1) + 2(t+1) + \cdots + m(t+1)) - n \\
          &= 2 + nm \\
          & = 2 + \frac{n^2}{t+1}.
    \end{align*}  
  Thus,
  $$
    \ccap(\Z^2_{R^2}) \le \limsup_{n \to \infty} \Big( 1 - \frac{|S|}{n^2} \Big) = \frac{t}{t+1}. 
  $$
\end{proof}
 
This theorem shows that for some recoverable systems, the two-dimensional capacity is not increased from the one-dimensional one: the maximally
sized system is obtained by stacking independent one-dimensional systems in the rows (or the columns). While the one-dimensional
components in the axial product provide an obvious lower bound on the capacity of the two-dimensional system, it is not
clear whether the equality holds in all cases. We could not find examples of axial products with higher capacity
by relying on both dimensions, and we leave this question as an open problem.

\section{Conclusion}
In this paper, we studied storage codes on finite graphs and their extensions to $\Z$ and $\Z^2$, which we
call recoverable systems. Earlier results for recoverable systems \cite{ElishcoBarg2020} relied on the assumption of shift invariance. Lifting this restriction enables one to connect it to the line of research on storage codes for finite graphs. Relying on resolvable designs, we proposed a new way of propagating storage codes by interleaving, which yields many new families from the known ones. Using bounds and constructions for finite storage codes, we found capacity of several recovery sets, such as $l_1$ or $l_\infty$ balls or cross-shaped regions. 
Finally, we established a link of the capacity problem to questions in additive combinatorics related to difference-avoiding sets, and found capacity values for some special recovery sets. 

There are numerous open problems that we leave for future research. To point out some of them, one may ask what is the maximum capacity of a two-dimensional system when only a subset of neighbors can be used for recovery. Next, it appears that the capacity of an axial product is always attained by stacking one-dimensional systems, although 
proving this has been elusive. If not true, then what are the conditions under which it is the case?  Among other questions: How does the fact that the order of the neighbors is known affect the capacity of the system?  And finally, is it possible to characterize the number of recovery functions needed to obtain the maximum capacity?

\appendix
\section{Linear programming bound}\label{sec:lp-bound}

To describe the linear programming bound of Mazumdar et al. \cite{MazMcgVor2019}, let us first define a $\tau$-cover by gadgets on a graph $G = (V, E)$.
For a subset of vertices $A\subset G$, let $\cl(A):=\{v : \cN(v) \subseteq A \}$ be the closure of $A$ which contains all vertices in $A$ and their neighbors.
A tuple $g = (S_1, S_2, c_1,c_2)$ is called a \emph{gadget} if 
\begin{enumerate}
  \item There exist two sets of vertices $A,B \subseteq V$ such that $S_1 = \cl(A) \cup \cl(B)$ and $S_2 = \cl(A) \cap \cl(B)$.
    We call $S_1$ the outside and $S_2$ the inside of the gadget.
  \item Colors $c_1$ and $c_2$ are picked from a fixed set of size $\tau$, and they are assigned to all the vertices in $S_1$, $S_2$ respectively. 
  Note that each vertex can have multiple assigned colors.
\end{enumerate}
We call $(\{v\}, \emptyset, c_1,c_2)$ a trivial gadget, and $w(g) = |A| + |B|$ the weight of gadget $g$. A collection of gadgets forms a $\tau$-cover if, for every color $c$, all the vertices with color $c$ form a vertex cover.
It was shown in~\cite[Theorem 8]{MazMcgVor2019} that the total weight of the gadgets that form a $\tau$-cover provides an upper bound on $\ccap(G)$. 
More formally, we have the following theorem.

\begin{theorem}[Linear programming (LP) bound~\cite{MazMcgVor2019}]\label{thm:LP}
  Let $\tau > 0$ be a fixed integer.
  For each gadget that contains $S \subseteq V$, define $\chi_{g,S} = \1\{g\text{ is involved in the cover}\}$, where $S$ is a part of $g$.
  Thus, each gadget $g$ corresponds to two variables $\chi_{g, S_1}$ and $\chi_{g, S_2}$, where $S_1$ and $S_2$ are the outside and the inside of $g,$ respectively.
  Denote by $c_g(S)$ the color of the vertex set $S$ in gadget $g$.

The capacity $\ccap(G)$ is bounded above by the solution to the following linear program:
  \begin{align}
    \text{minimize } \quad & \frac{1}{n\tau} \sum_{g,S} \chi_{g,S} \frac{w(g)}{2}  \nonumber \\
\text{s.t., } \quad & \frac{1}{\tau} \sum_{g,S} \chi_{g,S} \frac{w(g)}{2} \in \N \nonumber \\
\quad & \sum_{\substack{g,S: u \in S \\ C_g(S) = c}} \chi_{g,S} + \sum_{\substack{g',S': v \in S' \\ C_{g'}(S') = c}} \chi_{g',S'} \ge 1, \quad \forall (u,v) \in E, \forall c \label{eq:LP1} \\
& \chi_{g,S_1} = \chi_{g,S_2}, \quad \forall g \text{ with outside $S_1$ and inside $S_2$}. \label{eq:LP2}
  \end{align}

\end{theorem}

Note that conditions \eqref{eq:LP1} guarantee that the sets form a vertex cover, and~\eqref{eq:LP2} states that the inside is involved in the gadget cover if and only if the outside is also involved.

{\vspace*{.1in}\em Proof of Proposition~\ref{prop:LP-interleaved}:}
  Let $g = (S_1, S_2, c_1, c_2)$ be a gadget on $G$ and define $\bar{g} = (\bar S_1, \bar S_2,c_1,c_2),$ where $S_1 := \bigcup_{v \in S_1} \bar{v}$ and  $S_2 := \bigcup_{v \in S_2} \bar{v}$.
  If $g_1, \dots, g_m$ forms a $\tau$-cover of $G$, then it is straightforward to check that $\bar{g}_1, \dots, \bar{g}_m$ forms a $\tau$-cover of $\bar{G}$.
  Note that $w(\bar{g}) = s w(g)$, thus $\frac{1}{n\tau} w(g) = \frac{1}{ns\tau} w(\bar{g})$, and $\bar{g}_1, \dots, \bar{g}_m$ induce an upper bound that equals to the bound by  $g_1, \dots, g_m$, which is $\ccap(G)$ by our assumption.
  In other words, we have $\ccap_{q^k}(\bar{G}) \le \ccap(G)$.
  The matching lower bound follows by Proposition~\ref{prop:interleave-capacity}.



\begin{thebibliography}{10}
\providecommand{\url}[1]{#1}
\csname url@samestyle\endcsname
\providecommand{\newblock}{\relax}
\providecommand{\bibinfo}[2]{#2}
\providecommand{\BIBentrySTDinterwordspacing}{\spaceskip=0pt\relax}
\providecommand{\BIBentryALTinterwordstretchfactor}{4}
\providecommand{\BIBentryALTinterwordspacing}{\spaceskip=\fontdimen2\font plus
\BIBentryALTinterwordstretchfactor\fontdimen3\font minus
  \fontdimen4\font\relax}
\providecommand{\BIBforeignlanguage}[2]{{%
\expandafter\ifx\csname l@#1\endcsname\relax
\typeout{** WARNING: IEEEtranS.bst: No hyphenation pattern has been}%
\typeout{** loaded for the language `#1'. Using the pattern for}%
\typeout{** the default language instead.}%
\else
\language=\csname l@#1\endcsname
\fi
#2}}
\providecommand{\BIBdecl}{\relax}
\BIBdecl

\bibitem{AAK2001}
R.~Ahlswede, H.~Aydinian, and L.~Khachatryan, ``On perfect codes and related
  concepts,'' \emph{Designs, Codes and Cryptography}, vol.~22, no.~3, pp.
  221--237, 2001.

\bibitem{AK2015}
F.~Arbabjolfaei and Y.-H. Kim, ``Three stories on a two-sided coin: Index
  coding, locally recoverable distributed storage, and guessing games on
  graphs,'' in \emph{2015 53rd Annual Allerton Conference on Communication,
  Control, and Computing (Allerton)}, 2015, pp. 843--850.

\bibitem{AK2018}
------, ``Fundamentals of index coding,'' \emph{Foundations and Trends in Comm.
  Inf. Theory}, vol.~14, no. 3-4, pp. 163--346, 2018.

\bibitem{BarYossef2011}
Z.~Bar-Yossef, Y.~Birk, T.~S. Jayram, and T.~Kol, ``Index coding with side
  information,'' \emph{IEEE Trans. Inf. Theory}, vol.~57, no.~3, pp.
  1479--1494, 2011.

\bibitem{BEGY2022}
A.~Barg, O.~Elishco, R.~Gabrys, and E.~Yaakobi, ``Recoverable systems on lines
  and grids,'' in \emph{2022 IEEE International Symposium on Information Theory
  (ISIT)}, 2022, pp. 2637--2642.

\bibitem{barg2}
A.~Barg, M.~Schwartz, and L.~Yohananov, ``Storage codes on coset graphs with
  asymptotically unit rate,'' 2023, arxiv:2212.1217, [cs.IT].

\bibitem{BargZemor2021}
A.~Barg and G.~Z{\'e}mor, ``High-rate storage codes on triangle-free graphs,''
  \emph{IEEE Trans. Inf. Theory}, vol.~68, no.~12, pp. 7787--7797, 2022.

\bibitem{cap2008}
\BIBentryALTinterwordspacing
S.~Capobianco, ``Multidimensional cellular automata and generalization of
  {F}ekete's lemma,'' \emph{Discrete Mathematics \& Theoretical Computer
  Science}, vol.~10, 2008. [Online]. Available:
  \url{https://doi.org/10.46298/dmtcs.442}
\BIBentrySTDinterwordspacing

\bibitem{cap2022}
\BIBentryALTinterwordspacing
------, ``Fekete's lemma for componentwise subadditive functions of two or more
  real variables,'' \emph{Acta et Commentationes Universitatis Tartuensis de
  Mathematica}, vol.~26, no.~1, 2022. [Online]. Available:
  \url{https://doi.org/10.12697/ACUTM.2022.26.04}
\BIBentrySTDinterwordspacing

\bibitem{CollbornHand2007}
C.~J. Colbourn and J.~H. Dinitz, Eds., \emph{Handbook of combinatorial
  designs}, 2nd~ed., ser. Discrete Mathematics and its Applications (Boca
  Raton).\hskip 1em plus 0.5em minus 0.4em\relax Chapman \& Hall/CRC, Boca
  Raton, FL, 2007.

\bibitem{Delsarte1973}
P.~Delsarte, ``An algebraic approach to the association schemes of coding
  theory,'' \emph{Philips Res. Repts. Suppl.}, no.~10, 1973.

\bibitem{ElishcoBarg2020}
O.~Elishco and A.~Barg, ``Recoverable systems,'' \emph{IEEE Trans. Inf.
  Theory}, vol.~68, no.~6, pp. 3681--3699, 2022.

\bibitem{E11}
T.~Etzion, ``Product constructions for perfect {L}ee codes,'' \emph{IEEE Trans.
  Inf. Theory}, vol.~57, no.~11, pp. 7473--7481, 2011.

\bibitem{etzion2022perfect}
------, \emph{Perfect Codes and Related Structures}.\hskip 1em plus 0.5em minus
  0.4em\relax World Scientific, 2022.

\bibitem{Georgii2011}
H.-O. Georgii, \emph{Gibbs measures and phase transitions}, 2nd~ed.\hskip 1em
  plus 0.5em minus 0.4em\relax Berlin/New York: Walter de Gruyter \& Co., 2011.

\bibitem{GolovnevHaviv2022}
A.~Golovnev and I.~Haviv, ``The (generalized) orthogonality dimension of
  (generalized) {K}neser graphs: bounds and applications,'' \emph{Theory of
  Computing}, vol.~18, pp. 1--22, 2022, also: Proceedings of the 36th
  ComputationalComplexity Conference, 2021 (CCC?21).

\bibitem{gopalan2011locality}
P.~Gopalan, C.~Huang, H.~Simitci, and S.~Yekhanin, ``On the locality of
  codeword symbols,'' \emph{IEEE Trans. Inf. Theory}, vol.~58, no.~11, pp.
  6925--6934, 2011.

\bibitem{HuangXiang2023}
H.~Huang and Q.~Xiang, ``Construction of storage codes of rate approaching one
  on triangle-free graphs,'' manuscript, 2023.

\bibitem{Knuth1973}
D.~E. Knuth, \emph{The Art of Computer Programming. Vol. 1: Fundamental
  Algorithms}, 2nd~ed.\hskip 1em plus 0.5em minus 0.4em\relax Reading, MA:
  Addison-Wesley Pub. Co., 1973.

\bibitem{Kri2010}
F.~Krieger, ``The {O}rnstein-{W}eiss lemma for discrete amenable groups,'' Max
  Planck Institute for Mathematics, MPIM Preprint 2010?48, 2010.

\bibitem{Mazumdar2015}
A.~{Mazumdar}, ``Storage capacity of repairable networks,'' \emph{IEEE Trans.
  Inf. Theory}, vol.~61, no.~11, pp. 5810--5821, 2015.

\bibitem{MazMcgVor2019}
A.~Mazumdar, A.~McGregor, and S.~Vorotnikova, ``Storage capacity as an
  information-theoretic vertex cover and the index coding rate,'' \emph{IEEE
  Trans. Inf. Theory}, vol.~65, no.~9, pp. 5580--5591, 2019.

\bibitem{MazumdarWornell2015}
A.~Mazumdar, V.~Chandar, and G.~W. Wornell, ``Local recovery in data
  compression for general sources,'' in \emph{2015 IEEE International Symposium
  on Information Theory}, 2015, pp. 2984--2988.

\bibitem{ornstein1987entropy}
D.~S. Ornstein and B.~Weiss, ``Entropy and isomorphism theorems for actions of
  amenable groups,'' \emph{Journal d'Analyse Math{\'e}matique}, vol.~48, no.~1,
  pp. 1--141, 1987.

\bibitem{Rice2019}
A.~Rice, ``A maximal extension of the best-known bounds for the
  {F}urstenberg-{S}\'{a}rk\"{o}zy theorem,'' \emph{Acta Arith.}, vol. 187,
  no.~1, pp. 1--41, 2019.

\bibitem{Riis2007}
S.~Riis, ``Information flows, graphs and their guessing numbers,''
  \emph{Electron. J. Combin.}, vol.~14, no.~1, p. 17pp., 2007.

\bibitem{S78I}
A.~S{\'a}rk{\u o}zy, ``On difference sets of sequences of integers. {I},''
  \emph{Acta Mathematica Hungarica}, vol.~31, no. 1-2, pp. 125 -- 149, 1978.

\bibitem{Shanmugam2014}
K.~{Shanmugam} and A.~G. {Dimakis}, ``Bounding multiple unicasts through index
  coding and locally repairable codes,'' in \emph{2014 IEEE International
  Symposium on Information Theory}, 2014, pp. 296--300.

\bibitem{shi2022family}
M.~Shi, Y.~Xia, and D.~S. Krotov, ``A family of diameter perfect
  constant-weight codes from {S}teiner systems,'' 2022.

\bibitem{zhang2022linear}
T.~Zhang and G.~Ge, ``On linear diameter perfect {L}ee codes with diameter 6,''
  2022.

\bibitem{zhang2019nonexistence}
T.~Zhang and Y.~Zhou, ``On the nonexistence of lattice tilings of {${Z}^n$} by
  {L}ee spheres,'' \emph{Journal of Combinatorial Theory, Series A}, vol. 165,
  pp. 225--257, 2019.

\end{thebibliography}
\end{document}